\documentclass[12pt,reqno]{article}

\usepackage[usenames]{color}
\usepackage{amssymb}
\usepackage{amsmath}
\usepackage{amsthm}
\usepackage{amsfonts}
\usepackage{amscd}
\usepackage{graphicx}

\usepackage{tikz}
\usetikzlibrary{automata,positioning,arrows.meta,fit,decorations.pathmorphing}

\tikzset{
  initial text={},
  >=Stealth,
  every state/.style={minimum size=18pt, inner sep=2pt},
  node distance=2.2cm
}

\usepackage[colorlinks=true,
linkcolor=webgreen,
filecolor=webbrown,
citecolor=webgreen]{hyperref}

\definecolor{webgreen}{rgb}{0,.5,0}
\definecolor{webbrown}{rgb}{.6,0,0}

\usepackage{color}
\usepackage{fullpage}
\usepackage{float}

\usepackage{graphics}
\usepackage{latexsym}
\usepackage{epsf}
\usepackage{breakurl}

\newcommand{\seqnum}[1]{\href{https://oeis.org/#1}{\rm \underline{#1}}}
\DeclareMathOperator{\rep}{rep}

\DeclareMathOperator{\mex}{MeX}

\begin{document}

\theoremstyle{plain}
\newtheorem{theorem}{Theorem}
\newtheorem{corollary}[theorem]{Corollary}
\newtheorem{lemma}[theorem]{Lemma}
\newtheorem{proposition}[theorem]{Proposition}

\theoremstyle{definition}
\newtheorem{definition}[theorem]{Definition}
\newtheorem{example}[theorem]{Example}
\newtheorem{conjecture}[theorem]{Conjecture}

\theoremstyle{remark}
\newtheorem{remark}[theorem]{Remark}

\begin{center}
\vskip 1cm{\LARGE\bf
Variants of Wythoff game with terminal\\
\vskip .1in
positions or blocking maneuvers
}
\vskip 1cm
\large
Antoine Renard and Michel Rigo\\
Dept. of Mathematics, University of Li{\`e}ge\\
All{\'e}e de la d{\'e}couverte 12 (B37), B-4000 Li{\`e}ge, Belgium.
\href{mailto: Antoine.Renard@uliege.be}{\tt Antoine.Renard@uliege.be}\\
\href{mailto: M.Rigo@uliege.be}{\tt M.Rigo@uliege.be}\\
\end{center}

\vskip .2 in
\begin{abstract}
  We show how the software {\tt Walnut} can be used to obtain concise proofs of results concerning variants of the famous Wythoff game, in which blocking maneuvers or terminal positions are added, as discussed respectively by Larsson (2011) and Komak et al. (2025). Our approach provides automatic proofs that both confirm and extend their results, and the same techniques apply to newly introduced variants as well.
  
  Then, using classic techniques, we obtain new recursive and morphic characterizations of Wythoff-type games where the set of terminal positions $(x,y)$ satisfy $x+y\le \ell$. 

The use of {\tt Walnut} in combinatorial game theory is relatively recent, and only a few examples have been explored so far. The Wythoff game, being directly connected to the Fibonacci numeration system, proves especially well suited to this kind of approach. It permits us to solve instances for fixed value of a parameter.
\end{abstract}


\section{Introduction}

In combinatorial game theory, proofs of many results are obtained by using a suitable numeration system \cite{Duchene,Fraenkel1,Fraenkel2,RigoLNM}. As a typical example, the Zeckendorf numeration based on the Fibonacci sequence is used to characterize the $\mathcal{P}$-positions (i.e., losing positions) of the Wythoff game \cite{Fraenkel3}. Fraenkel showed that a pair $(a,b)$ of integers such that $a\le b$ is a $\mathcal{P}$-position of Wythoff game if and only if $\rep_F(a)$ ends with an even number of zeroes and $\rep_F(b)$ is the left-shift of $\rep_F(a)$, i.e., $\rep_F(b)=\rep_F(a)0$.

When the numeration system has good properties (i.e., it is an {\em addable} system as discussed in \cite{ShallitBook}), the formalism of first-order logic can be used to express certain properties of combinatorial games. We have recently exploited this approach and used {\tt Walnut} to reprove classical results in combinatorial games theory, as well as to obtain new ones \cite{MRRW}. {\tt Walnut} is a free, open-source prover for first-order statements dealing with automatic sequences \cite{Mousavi,ShallitBook}. 

We assume that the reader is sufficiently familiar with the use of {\tt Walnut}. For references on combinatorial game theory, we refer to \cite{RigoLNM} and \cite{MRRW}.

The Wythoff game is a {\em take-away} game: two players alternately remove tokens from two piles. One can remove a positive number of token from one pile or, remove the same positive number of token from both piles. The Wythoff game is {\em impartial}: the set of options (i.e., positions reachable in one legal move) of a given position does not depend on which player is in turn to move. Wythoff game can be seen as a queen on an infinite chessboard $\mathbb{N}^2$: in a position $(x,y)\neq(0,0)$ the queen can move horizontally to $(x-n,y)$ with $0<n\le x$, vertically to $(x,y-n)$ with $0<n\le y$ or diagonally to $(x-n,y-n)$ with $0<n\le \min\{x,y\}$. According to the {\em normal convention}, a player who is unable to move loses the game; passing is not allowed. Otherwise stated, the player who takes the last token wins or, in an equivalent way, when the queen reaches $(0,0)$.

In this article, we are mainly interested in two kinds of variations of the Wythoff game. 

In \cite{Komak}, a variant of Wythoff game has been introduced: the set $\{(x,y)\mid x+y\le 2\}$ is declared to be the set of {\em terminal} positions. If a player moves the queen into this terminal set with the usual Wythoff moves, that player wins the game. Let us call this game $K^2$. We chose this notation because we will later introduce a generalization of the game obtained by varying the set of terminal positions. Let $\phi$ be the golden ratio~$(1+\sqrt{5})/2$. We first define a sequence~$(g(n))_{n\ge 0}$ as follows.
\begin{definition}\label{def:fg}
  We let $g(0)=1$, $g(1)=0$ and
  \[
    g(n)=\left\{
      \begin{array}{ll}
        1-g(m),& \text{ if }\lfloor n\phi\rfloor= \lfloor m(\phi+1)\rfloor+1 \text{ for some }m\ge 0;\\
        1,& \text{ otherwise}.
      \end{array}\right.
  \]
  
\end{definition}
The first values of the sequence~$g$ are given in Table~\ref{tab:gh}. 
\begin{table}[h!tb]
  \centering
  \[
    \begin{array}{c|ccccccccccccccccccccc}
      & 0 & 1 & 2 & 3 & 4 & 5 & 6 & 7 & 8 & 9 & 10  & 11 & 12 & 13 & 14 & 15 & 16 & 17 & 18 &  19 & 20 \\
      \hline
      h& 0 & 1 & 1 & 2 & 3 & 3 & 4 & 4 & 5 & 6 & 6 &  7 & 8 & 8 & 9 & 9 & 10 & 11 & 11 & 12 &  12 \\
      g&1 & 0 & 1 & 1 & 0 & 0 & 1 & 1 & 1 & 1 & 0 &  1 & 0 & 0 & 1 & 0 & 1 & 1 & 0 & 1 & 1 \\
    \end{array}
  \]
  \caption{First values of the sequence $g$ and Hofstadter sequence $h$.}
  \label{tab:gh}
\end{table}

The main result of Komak {\em et al.} provides an algebraic description of set of $\mathcal{P}$-positions of $K^2$:
\begin{theorem}\label{thm:K2}
  The set of $\mathcal{P}$-positions of the variant of Wythoff game where the set of terminal positions is $\{(x,y)\mid x+y\le 2\}$,  is exactly \begin{equation}\label{eq:Ppos}
  \left\{ (\lfloor n\phi\rfloor +g(n)-1, \lfloor n \phi^2\rfloor +g(n))\mid n\ge 0\right\} \cup
  \left\{ (\lfloor n \phi^2\rfloor +g(n), (\lfloor n\phi\rfloor +g(n)-1)\mid n\ge 0\right\}
\end{equation}
where $g:\mathbb{N}\to\mathbb{N}$ is the function given in Definition~\ref{def:fg}.
\end{theorem}

The authors provide a classical proof of their result by showing that the set of $\mathcal{P}$-positions described by \eqref{eq:Ppos} is both stable and absorbing (see Definition~\ref{def:stab}). The proof is quite long and requires a detailed case analysis. Without diminishing their achievement, once the set \eqref{eq:Ppos} has been conjectured, if the set can be handled by {\tt Walnut}, then it becomes straightforward to produce an automated proof in just a few lines. Our approach therefore complements the work carried out in \cite{Komak}.

In \cite{block}, other variants of Wythoff game are studied, in which some blocking maneuvers are added. Let $k\ge 1$. In the game denoted by $W^k$, for each move, before the next player (i.e., the player who is about to play) moves, the previous player (i.e., the one who has just played) may declare at most $k-1$ of the options as forbidden. When the next player has moved, any blocking maneuver is forgotten and has no further incidence on the game. This terminology of previous and next player explains why we speak of $\mathcal{N}$- and $\mathcal{P}$-positions, respectively. For $k=1$, this is the classical game of Wythoff. When $k=2,3$, we reconsider Larsson's result and obtain an automated proof, replacing a 2-page-long analysis by some easy-to-describe first-order formulas. 

As Shallit has already noted on several occasions, the two approaches are complementary. A classical, purely combinatorial proof usually provides structural insights, whereas the use of automated provers helps avoid lengthy case analyses and enables the exploration of directions that are difficult to access through traditional techniques. In particular, in this article, we easily obtain a variety of results that would otherwise require significantly more effort using ``classical'' methods. Moreover, these results lead us to formulate general theorems.
\smallskip

This article is organized as follows. Section~\ref{sec:2} --- which was the starting point of this article --- aims to give an automatic proof of the algebraic characterization \eqref{eq:Ppos} of the $\mathcal{P}$-positions of $K^2$. We begin with preliminary results about Fibonacci-automatic sequences. We define a notion of $\varphi$-morphism and describe a heuristic that we extensively use throughout the paper to construct such $\varphi$-morphisms. Given a sufficiently long prefix of an infinite word, we obtain morphisms that can be used in our automatic proofs (and we can therefore prove the correctness of the procedure). Basic results from combinatorial game theory are given in Section~\ref{ssec:autom_proof}.

In Section~\ref{sec:newGame}, we obtain results about new games: a parameterized version of a variant of Wythoff game denoted by $K^\ell$, where the set of terminal states is
\[
  \{(x,y)\in\mathbb{N}^2\mid x+y\le \ell\}
\]
for some $\ell\in\mathbb{N}$. We study the games $K^1$, $K^3$ and $K^4$ and obtain algebraic characterizations of the $\mathcal{P}$-positions similar to \cite{Komak}. In particular, as for the classical Wythoff game $K^0$, the game $K^1$ has a nice set of $\mathcal{P}$-positions: a pair $(a,b)$ of integers such that $a\le b$ is a $\mathcal{P}$-position of $K^1$ if and only if its Fibonacci representation $\rep_F(a)$ ends with zero and $\rep_F(b)$ is of the form $\rep_F(a)1$, see Theorem~\ref{thm:carK1}. 

In Section~\ref{sec:6}, we go further in the analysis of the games $K^\ell$ for an arbitrary $\ell\ge 1$. We provide a recursive characterization of the set of $\mathcal{P}$-positions. The non-terminal $\mathcal{P}$-positions $(a_n,b_n)_{n\ge 0}$ of $K^\ell$ are given, in Theorem~\ref{thm:rec_car}, by
\[
  (a_0,b_0)=(\ell+1,2\ell+2) \quad\text{ and }\quad \forall n\ge 1, \left\{
    \begin{array}{l}
a_n =\mex (\{ a_i,b_i \mid i<n \}\cup\{0,1,\ldots,\ell\}) \\
b_n =a_n+n+\ell+1.\\
    \end{array}
  \right.
\]
We prove this result in an automatic way for $\ell=1,2$, then extend it for all values of the parameter using classical arguments. With Corollary~\ref{cor:behave}, we generalize Theorem~\ref{thm:K2} and obtain that $(a_n,b_n)_{n\ge 0}$ are of the form
\[ a_n=\lfloor (n+\ell)\phi\rfloor + \lambda_\ell(n) \quad\text{ and }\quad
  b_n=\lfloor (n+\ell)\phi^2\rfloor + \lambda_\ell(n)+1\]
for some bounded integer-valued function $\lambda_\ell$. Our result can be related to the recent paper \cite{Letouzey} where discrepancy of generalized Hofstadter functions is studied (and formalized in {\tt Coq}). Our methods are inspired by   \cite{Fraenkel5,Kimberling2,Larsson}.

It is well-known that the set of $\mathcal{P}$-positions of the Wythoff game are coded by the Fibonacci word~$\mathbf{f}=abaababa\cdots$, see \cite{Vision,Duchene}. If we start indexing letters of~$\mathbf{f}$ with $1$, the indices of the letters $a$ (resp. $b$) in~$\mathbf{f}$ correspond to the sequence $(a_n)_{n\ge 0}$ (resp. $(b_n)_{n\ge 0}$). Using the recursive definition from the previous section, we obtain  similar results with $K^1,K^2,K^3$: their $\mathcal{P}$-positions are encoded by morphic words obtained by $\varphi$-morphisms and codings. 

In Section~\ref{sec:3}, we obtain automatic proofs of Larsson's result about $W^2$ and $W^3$. Interestingly, the automata arising in the computations are much larger than those for $K^2$.

In Section~\ref{sec:redundant}, we are concerned with redundant moves of the games discussed so far. A move is considered {\em redundant} if, upon its removal from the rule-set, then the set of $\mathcal{P}$-positions remains unchanged. It is a classical question in combinatorial game theory: do distinct rule-sets yield the same set of $\mathcal{P}$-positions? In particular, can certain moves be added or removed without altering the $\mathcal{P}$-positions?

\begin{remark}
  A {\tt Jupyter} notebook recording all the {\tt Walnut} computations is available (it has been produced using Ollinger's {\tt Walnut-Kernel}. 
  The {\tt Mathematica} code and the various {\tt Walnut} files are also available online\footnote{at {\tt https://hdl.handle.net/2268/338482}, you may also download a standalone Wolfram Player.}. In particular, the package permits us to compute $\mathcal{P}$-positions and $\varphi$-morphisms.
\end{remark}
\section{Our starting point: The game $K^2$}\label{sec:2}
This section is about the game $K^2$. The aim is to give an automatic proof of the characterization of its $\mathcal{P}$-positions. We first show that the sequence $g$ appearing in \eqref{eq:Ppos} is {\em Fibonacci-automatic}, i.e., for all $n\ge 0$, its $n$th term is the output of a DFAO fed with the Fibonacci representation of $n$ (see Definition~\ref{def:fib_rep}). To discover this DFAO we apply a heuristic that we describe in Section~\ref{ssec:heuristic} (and then applied several times in Section~\ref{sec:newGame}). Finally, we give the automatic proof. To be self-contained, classical definitions from combinatorial game theory are given in Section~\ref{ssec:autom_proof}.

\subsection{A first result on the sequence $g$}
To make use of {\tt Walnut}, all the ingredients need to be definable in a suitable extension of $\langle\mathbb{N},+\rangle$. Here, we show that the sequence $g$ can be computed with a DFAO fed with Fibonacci representations. 

Let us follow the presentation given in \cite{Komak} where an alternative definition of $g$ is provided.
\begin{proposition}\label{prop:h}
  We have $g(0)=g(1)=1$ and for all $n\ge 2$, 
\[
  g(n) =
  \left\{\begin{array}{ll}
                  1- g(h(n-1)),& \text{ if }h(n-2)<h(n-1);\\
                  1,& \text{ otherwise}.\\
         \end{array}\right.
\]
where $h$ is the {\em Hofstadter $G$-sequence} \seqnum{A005206}, see \cite{Hofstadter}.   
\end{proposition}
A characterization of the entries of Hofstadter $G$-sequence in terms of the lower and upper Wythoff sequence is given in \cite{Dekking}.

\begin{definition}\label{def:fib_rep}
Let $F=(F_n)_{n\ge 0}$ be the Fibonacci sequence where $F_0=1$, $F_1=2$ and $F_{n+2}=F_{n+1}+F_n$ for all $n\ge 0$. 
We let $\rep_F(n)$ denote the {\em greedy representation} of the integer $n>0$, that is the unique word $d_k\cdots d_0$ over $\{0,1\}$ such that
\[
  n=\sum_{i=0}^k d_iF_i
\]
with $d_k\neq 0$ and $d_{i+1}d_i\neq 11$ for $i<k$. The representation of $0$ is $\rep_F(0)=\varepsilon$ the empty word. 
\end{definition}

The Beatty sequences $\lfloor n\phi\rfloor$ and $\lfloor n\phi^2\rfloor$, also known respectively as lower and upper Wythoff sequences \seqnum{A000201}, \seqnum{A001950} can be defined in {\tt Walnut} using 
\begin{verbatim}
reg shift {0,1} {0,1} "([0,0]|([0,1][1,1]*[1,0]))*":
def phin "?msd_fib (s=0 & n=0) | Ex $shift(n-1,x) & s=x+1":
def phi2n "?msd_fib (s=0 & n=0) | Ex,y $shift(n-1,x) & $shift(x,y) & s=y+2":
\end{verbatim}
For instance, the command
\begin{verbatim}
eval test "?msd_fib $phin(3,4)":
\end{verbatim}
returns {\tt TRUE} because $\lfloor 3\phi\rfloor=4$. See \cite[p. 278]{ShallitBook}.

\begin{proposition}\label{pro:fibaut}
  The sequence $(g(n))_{n\ge 0}$ is Fibonacci-automatic: The DFA depicted in Fig.~\ref{fig:automG} accepts $\rep_F(n)$ if and only if $g(n)=1$.
\end{proposition}

\begin{figure}[h!tb]
  \centering
  \scalebox{.7}{\includegraphics{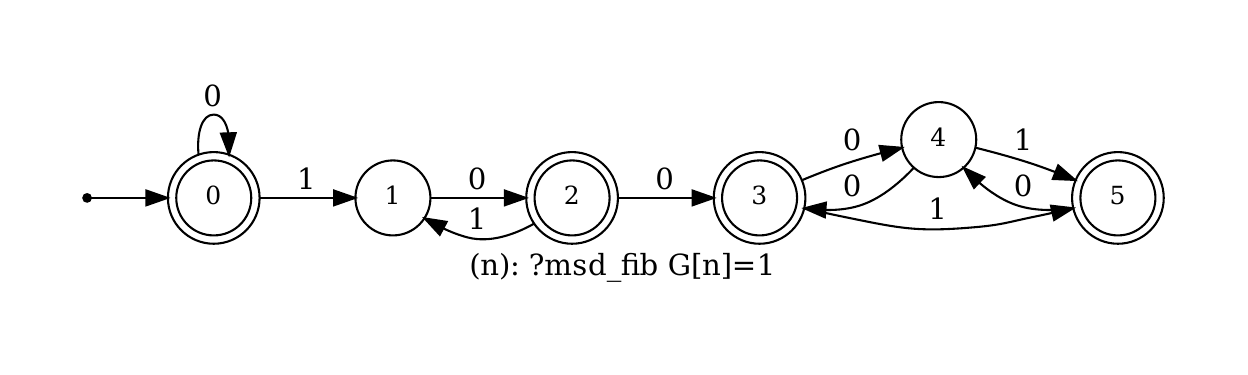}}
  \caption{A DFA recognizing $\{\rep_F(n):g(n)=1\}$.}
  \label{fig:automG}
\end{figure}

\begin{proof}
  The candidate automaton depicted in Figure~\ref{fig:automG} can be encoded and stored in {\tt Word Automata Lib}. So we have access in {\tt Walnut} to a candidate function {\tt G}. A discussion about how we have obtained this DFA is given in Section~\ref{ssec:heuristic}. Now, it is well-known that the Hofstadter sequence satisfies, for $n>0$, 
  \[
    h(n)=\biggl\lfloor \frac{n+1}{\phi}\biggr\rfloor=\lfloor (n+1)\phi\rfloor -n-1.
  \]
See, \cite{Gault,Granville}. Hence, $h(n)=y$ can be easily coded as a binary predicate $h(n,y)$ by
\begin{verbatim}
def hofstadter "?msd_fib (y=1 & n=0) | Ex $phin(n+1,x) & y+n+1=x":
\end{verbatim}
Now we have to check that the automaton satisfies the properties expressed in Proposition~\ref{prop:h}
\begin{verbatim}
eval test "?msd_fib An n>=2 => (Es,t ($hofstadter(n-2,s) & $hofstadter(n-1,t)
     & (s<t => G[n]+G[t]=1)))":
eval test2 "?msd_fib An n>=2 => (Es,t ($hofstadter(n-2,s) & $hofstadter(n-1,t)
     & (s=t => G[n]=1)))": 
\end{verbatim}
and both commands return {\tt True}.  
\end{proof}

\subsection{A heuristic generating $\varphi$-morphisms}\label{ssec:heuristic}
We present here a heuristic that, given a (long enough) prefix of an infinite word~$\mathbf{w}$ being the coding of a fixed point of what we call a $\varphi$-morphism, suggests a morphism~$\mu$ and a coding~$\rho$ such that $\mathbf{w}=\rho(\mu^\omega(0))$. This heuristic will be used many times (it was already used in \cite{DFNR} but not explicitly described). Once the morphism is obtained, we can then prove its correctness with {\tt Walnut}.
\begin{definition}
  A morphism $\mu:A^*\to A^*$ is said to be a {\em $\varphi$-morphism}, if for any with fixed point $\mathbf{w}=w_0w_1\cdots$ of $\mu$,
  \begin{itemize}
  \item if $\rep_F(n)$ ends with $1$, then $\mu(w_n)=w_{\rep_F(n)0}$,
  \item if $\rep_F(n)$ ends with $0$, then $\mu(w_n)=w_{\rep_F(n)0}w_{\rep_F(n)1}$.
  \end{itemize}
\end{definition}

We have chosen a simplified presentation, adapted to the situation considered here. Indeed, this is a special case of what is called a $U$-substitution in \cite{BH} for a wide class of numeration systems. Let $\sigma:a\mapsto ab$, $b\mapsto a$ the usual Fibonacci morphism generating the Fibonacci word $\mathbf{f}=abaababaa\cdots$. If $\mu:A^*\to A^*$ is a $\varphi$-morphism, then there is a coding $f:A^*\to\{a,b\}^*$ such that
\begin{equation}\label{eq:morphf}
  f(\mu(i))=\sigma(f(i)),\text{ for all } i\in A.
\end{equation}
Roughly speaking, $\mu$ is defined on a larger alphabet than $\{a,b\}$ but it preserves the underlying structure of the morphism~$\sigma$.

\begin{example}\label{exa:mu}
  The morphism
  \[
    \mu:0 \mapsto 01,\, 1 \mapsto 2,\, 2 \mapsto 31,\, 3 \mapsto 45,\, 4 \mapsto 35,\, 5 \mapsto 4
  \]
  is a $\varphi$-morphism. Consider the coding $f:0,2,3,4\mapsto a$ and $1,5\mapsto b$.
\end{example}

\begin{definition}\label{def:dfao}
  We will often associate a DFA with a morphism $\mu:A^*\to A^*$. The set of states is the alphabet~$A$, the alphabet of the automaton is $0,1$. This is a classical construction that goes back to Cobham \cite{cobham} and it corresponds to the {\tt promote} operation in {\tt Walnut}. The transitions of the DFA are given by the morphism: if $\mu(c)=de$, $c,d,e\in A$, then there are edges from $c$ to $d$ and $e$ with respective labels $0$ and $1$. If $\mu(c)=d$, $c,d\in A$, then there is a single edge from $c$ to $d$ labeled by $0$. If we have an extra coding $\rho:A\to B$ (i.e., a letter-to-letter morphism), then the DFA can be turned into a DFAO where the output function is exactly~$\rho$.
\end{definition}

With the morphism from Example~\ref{exa:mu}, the corresponding automaton is depicted in Figure~\ref{fig:muf}. Property~\eqref{eq:morphf} means that there is morphism of automata between the one associated with $\mu$ and the one associated with the Fibonacci morphism: if there is a transition between $c$ and $d$ with label $i$, then there is a transition between $f(c)$ and $f(d)$ with the same label $i$.

\begin{figure}[h!tb]
  \centering
  \begin{tikzpicture}[auto]
  \node[state, initial]   (q1) {0};
  \node[state] (q3) [right=1cm of q1] {2};
  \node[state] (q4) [right=1cm of q3] {3};
  \node[state] (q5) [right=of q4] {4};
  \node[state, initial] (a) [right=3cm of q5] {$a$};
  
  \node[state] (q2) [below=2cm of q1] {1};
  \node[state] (q6) [right=4cm of q2]     {5};
  \node[state] (b) [below=2cm of a] {$b$};

\node[draw, dashed, rounded corners,
        fit=(q1) (q3) (q4) (q5),
        inner xsep=6mm, inner ysep=6mm, yshift=3mm, xshift=-2mm] (topbox) {};
        
\node[draw, dashed, rounded corners,
        fit=(q2) (q6),
        inner xsep=6mm, inner ysep=3mm] (topbox) {};
        
  \path[->]
    (q1) edge[loop above]  node {0} (q1)
    (q1) edge              node[left] {1} (q2)
    (q2) edge              node {0} (q3)
    (q3) edge              node {0} (q4)
    (q3) edge[bend left]   node[pos=0.2] {1} (q2)
    (q4) edge[bend left]   node {0} (q5)
    (q4) edge[bend right]  node[left] {1} (q6)
    (q5) edge              node {0} (q4)
    (q5) edge[bend left]   node {1} (q6)
    (q6) edge              node {0} (q5);

  \path[->]
    (a) edge[loop above]  node {0} (a)
    (a) edge[bend right]   node[left] {1} (b)
    (b) edge[bend right]  node[right] {0} (a);

    \draw[->, shorten <=5mm, shorten >=8mm]
    (q5) -- node[above] {$f$} (a);
        \draw[->, shorten <=7mm, shorten >=2mm]
  (q6) -- node[above] {$f$} (b);
  \end{tikzpicture}
  \caption{Automata associated with $\mu$ and the Fibonacci morphism.}
  \label{fig:muf}
\end{figure}
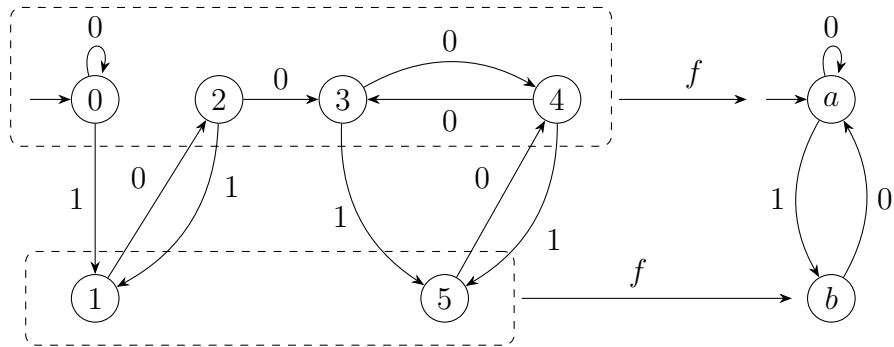
\bigskip

{\bf Heuristic}.
We describe a useful heuristic to devise the automaton in Figure~\ref{fig:automG} (and many other ones in the article). Assume that the sequence $g$ in Table~\ref{tab:gh} is the fixed point of a $\varphi$-morphism~$\mu$. In this case, we start from a single letter $g_n$, $n\ge 0$, then iterate the morphism a finite number of times (here, three times) and record the different factors obtained. Since we assume that $\mu$ is a $\varphi$-morphism, we know precisely where the factors $\mu^i(g_n)$ are located.

Precisely, if $\rep_F(n)=u$, consider the Fibonacci representations of length $|u|+i$ having $u$ as a prefix. These words give the positions of the letters in $\mu^i(g_n)$. Let $\alpha(i,n)$ and $\beta(i,n)$ be the first and last positions (Figure~\ref{fig:type} may help the reader). So, formally, a $t$-type is a $(t+1)$-tuple made of factors in prescribed positions depending only on $n$
\[
  (g_n,\ g_{\alpha(1,n)}\cdots g_{\beta(1,n)},\ \ldots,\  g_{\alpha(t,n)}\cdots g_{\beta(t,n)}).
\]

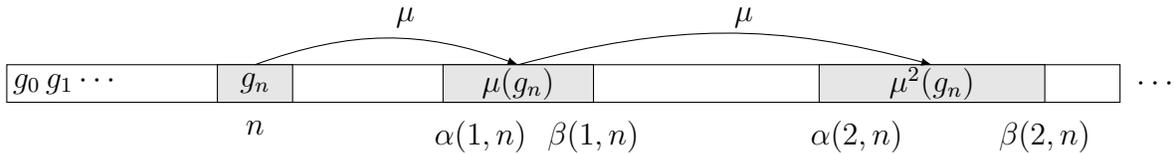
\begin{figure}[h!tb]
  \centering
\begin{tikzpicture}[>=latex]

  \draw (-.8,0) rectangle (14,.5);

  \node at (0,0.25) {$g_0\, g_1\cdots$};
  \node at (14.5,0.25) {$\cdots$};

  \draw[fill=gray!20] (2,0) rectangle (3,.5);
  \node at (2.5,0.25) {$g_n$};
  \node[below=3pt] at (2.5,0) {$n$};

  \draw[fill=gray!20] (5,0) rectangle (7,.5);
  \node[below=3pt] at (5.5,0) {${\alpha(1,n)}$};
  \node[below=3pt] at (7,0) {${\beta(1,n)}$};
  \node at (6,0.25) {$\mu(g_n)$};
   
  \draw[fill=gray!20] (10,0) rectangle (13,.5);
   \node at (11.5,0.25) {$\mu^2(g_n)$};
  \node[below=3pt] at (10.5,0) {${\alpha(2,n)}$};
  \node[below=3pt] at (13,0) {${\beta(2,n)}$};
  \coordinate (letter) at (2.5,.5);
  \coordinate (factor) at (6,.5);
  \coordinate (factor2) at (11.5,.5);
  
  \draw[->,bend left=20] (letter) to (factor);
   \node[above=3pt] at (4.5,.75) {$\mu$};
 \draw[->,bend left=15] (factor) to (factor2);
 \node[above=3pt] at (9,.75) {$\mu$};
\end{tikzpicture}
 \caption{Construction of a type.}
  \label{fig:type}
\end{figure}

Over a long prefix, we observe only six distinct $3$-types (and this number remains stable when the number of iterations increases). As an example, starting from $g_4=0$, since $\rep_F(4)=101$, this word can be extended to $1010$ whose value is $7$ and $g_7=1$. Now $1010$ can be extended to $10100$ and $10101$ giving $g_{11}g_{12}=10$. Finally these two Fibonacci representations can again be extended to $101000$, $101001$ and $101010$ giving $g_{18}g_{19}g_{20}=011$. This is the second line in Table~\ref{tab:types}. The reader may notice that the six elements in the column $\mu^2(g_n)$ are pairwise distinct. The last column does not bring more information to distinguish the rows of the table. So on this example, considering $2$-types would have been enough. However, $1$-types are not enough: we cannot distinguish the first and last rows. 
\begin{table}[h!tb]
\[
  \begin{array}{c|llll}
    &g_n & \mu(g_n) & \mu^2(g_n) & \mu^3(g_n) \\
    \hline
 0&1& 0& 1& 10\\
 1&0& 1& 10& 011\\
 2&1& 10& 011& 11010\\
 3&1& 01& 110& 01011\\
 4&0& 11& 010& 11011\\
 5&1& 0& 11& 010\\
  \end{array}
\]
  \caption{The different $3$-types when applying a $\varphi$-morphism to $g$.}
  \label{tab:types}
\end{table}

Now we associate each symbol $g_n$ with its $3$-type and get the infinite word
\[  
  01|2|31|45|2|35|4|31|45|4|35|45|\cdots
\]
where the vertical bars indicate the factorization of the Fibonacci word $ab|a|ab|ab|a|\cdots$ with factors $\sigma(a)=ab$ and $\sigma(b)=a$, the images of the letters by the Fibonacci morphism. From this, we deduce the morphism $\mu$ given in Example~\ref{exa:mu}: the image of the $n$th letter has to be the $n$th block in the factorization. So $0\mapsto 01$, then $1\mapsto 2$ and so on. From the morphism, it is a routine to obtain the DFA in Figure~\ref{fig:automG}, see Definition~\ref{def:dfao}.

Finally, one may define a coding~$\rho$ to recover the original sequence where each type is mapped to the first element of the tuple. So, $0,2,3,5\mapsto 1$ and $1,4\mapsto 0$. This can equivalently be performed by comparing the morphic word and the original sequence
\begin{align*}
  &01231452354314543545\cdots\\
  &10110011110100101101\cdots
\end{align*}
and matching corresponding elements.
\subsection{An automatic proof for $K^2$}\label{ssec:autom_proof}

Now that $g$ has been defined by a DFA fed with Fibonacci representations (and stored as {\tt G.txt} in {\tt Word Automata Lib}), the set of $\mathcal{P}$-positions \eqref{eq:Ppos} obtained by Komak et al. can be described by the following commands:
\begin{verbatim}
def pposK2_asym "?msd_fib (a+b<=1) | En,x $phin(n,x) & a+1=x+G[n]
                                                     & b=x+n+G[n]":
def pposK2 "?msd_fib $pposK2_asym(a,b) | $pposK2_asym(b,a) ":
\end{verbatim}
The resulting automaton has $27$ states.

For the sake of presentation, let us recall some basics on combinatorial games.
\begin{definition}\label{def:stab}
A set $S$ of positions is {\em stable} if, for all $s,t\in S$, $s\neq t$, there is no move between $s$ and $t$. A set $S$ of positions is {\em absorbing}, if for all $t\not\in S$, there exists $s\in S$ such that there is a move from $t$ to $s$ (or, $t$ has an option in $S$).  
\end{definition}
It is a well-known property in graph theory that every acyclic graph has a unique kernel, i.e., a stable and absorbing set of vertices. In a take-away game where tokens are removed from piles, the corresponding game graph is acyclic: it is not possible to visit twice the same position during a game.
\begin{proposition}[Folklore]\label{pro:kernel}
  The sets of $\mathcal{P}$- and $\mathcal{N}$-positions of an
  impartial acyclic game are uniquely determined by the following two
  properties:
\begin{enumerate}
\item Every move from a $\mathcal{P}$-position leads to an
  $\mathcal{N}$-position; equivalently there is no move between two $\mathcal{P}$-positions (stability property of $\mathcal{P}(G)$).
  \item From every $\mathcal{N}$-position, there exists a move leading to a $\mathcal{P}$-position (absorbing property of $\mathcal{P}(G)$).
\end{enumerate}
\end{proposition}

As observed in \cite{MRRW}, it is then a routine test to check stability and absorbing property:
\begin{verbatim}
eval stableK2 "?msd_fib Ap,q,r,s (($pposK2(p,q) & $pposK2(r,s)
     & p>=r & q>=s & p+q>2) => ((p=r & q=s) | (p>r & q>s & p+s!=q+r)))":

eval absorbingK2 "?msd_fib Ap,q (~$pposK2(p,q) => Ex,y
     (x<=p & y<=q & $pposK2(x,y) & (p+y=q+x | p=x | q=y))) ":
\end{verbatim}
These two commands return {\tt TRUE} resulting in the following result, \cite[Thm.~3]{Komak}. From Proposition~\ref{pro:kernel}, we immediately get an alternative proof of Theorem~\ref{thm:K2}.

\section{Towards an algebraic characterization for $K^\ell$}\label{sec:newGame}

In \cite{Komak}, Komak et al. chose $\{(x,y)\mid x+y\le 2\}$ as set of terminal positions. It is natural to consider a parameterized version where the set of terminal states is
\[
  \{(x,y)\in\mathbb{N}^2\mid x+y\le \ell\}
\]
for some $\ell\in\mathbb{N}$. The corresponding game is denoted by $K^\ell$. Note that the game $K^0$ is the classic Wythoff game. In this section, we obtain new results concerning the games $K^1$ and~$K^3$.

\subsection{The $K^1$-game}
Up to our knowledge, the game $K^1$ does not seem to have been studied in the literature. Computing the first $\mathcal{P}$-positions of $K^1$ and writing their Fibonacci expansions permitted us to state a conjecture that can be tested with {\tt Walnut}. This result is quite similar to the famous characterization of the $\mathcal{P}$-positions of Wythoff game obtained by Fraenkel \cite{Fraenkel3}.
\begin{table}[h!t]
  \[
    \begin{array}{c|r|c|r}
      a_n & \rep_F(a_n) & b_n & \rep_F(b_n) \\
      \hline
 0&\varepsilon  & 1 &1 \\
 2&10 & 4 &101 \\
 3&100 & 6 & 1001 \\
 5&1000 & 9 & 10001 \\
 7&1010 & 12 & 10101 \\
 8&10000 & 14 & 100001 \\
 10&10010 & 17 &100101 \\
 11&10100 & 19 & 101001 \\
 13&100000 & 22 & 1000001 \\
 15&100010 & 25 & 1000101 \\
  \end{array}
\]
  \caption{Representations of the first $\mathcal{P}$-positions in $K^1$.}\label{tab:K1}
\end{table}

\begin{theorem}\label{thm:carK1}
  A position $(a,b)$, with $a\le b$, is a $\mathcal{P}$-position of the game $K^1$, with a set of terminal positions being $\{(x,y)\mid x+y\le 1\}$ and Wythoff's moves, if and only if $\rep_F(a)$ ends with $0$ and $\rep_F(b)=\rep_F(a)1$. 
\end{theorem}
The candidate set of $\mathcal{P}$-positions can be readily implemented.
\begin{verbatim}
reg end_zero msd_fib "(0|1)*0":
reg add_one {0,1} {0,1} "([0,0]|([0,1][1,1]*[1,0]))*[0,1]":
def pposK1_asym "?msd_fib (a=0 & b=0) | ($end_zero(a) & $add_one(a,b))":
def pposK1 "?msd_fib $pposK1_asym(a,b) | $pposK1_asym(b,a)":
\end{verbatim}
The corresponding DFA is depicted in Figure~\ref{fig:automK1}.
\begin{figure}[h!tb]
  \centering
  \scalebox{.7}{\includegraphics{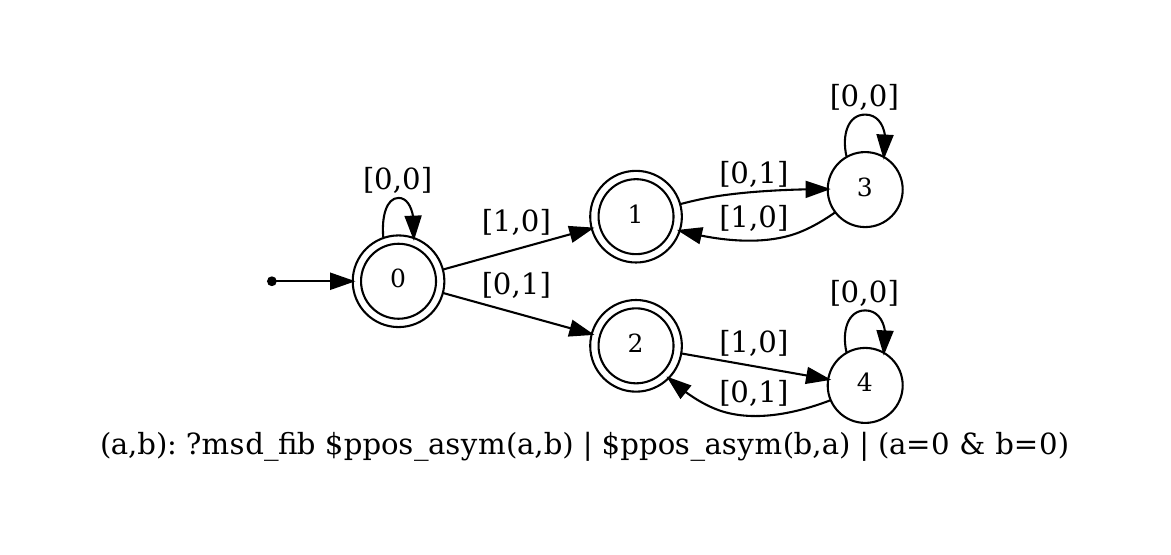}}
  \caption{A DFA recognizing $\mathcal{P}$-positions of $K^1$.}
  \label{fig:automK1}
\end{figure}
The sequence $(a_n)_{n\ge 0}=0, 2, 3, 5, 7, 8, 10,\ldots$ (first component of the $\mathcal{P}$-positions) is referred to the Fibonacci-even numbers as \seqnum{A022342} (as binary expansion of even integers end with zero). The complementary sequence $(b_n)_{n\ge 0}$ appears in the OEIS as \seqnum{A003622}.

\begin{proof}
The proof is a routine procedure. One has to check that the candidate set is stable and absorbing:
\begin{verbatim}
eval stableK1 "?msd_fib Ap,q,r,s (($pposK1(p,q) & $pposK1(r,s) & p>=r & q>=s
      & p+q>1) => ((p=r & q=s) | (p>r & q>s & p+s!=q+r)))":
eval absorbingK1 "?msd_fib Ap,q (~$pposK1(p,q) => Ex,y (x<=p & y<=q
      & $pposK1(x,y) & (p+y=q+x | p=x | q=y)))":
\end{verbatim}
Both commands evaluate to {\tt True}. The conclusion follows from Proposition~\ref{pro:kernel}.
\end{proof}

\subsection{The $K^3$-game}
The case of $K^3$ is more interesting. In Figure~\ref{fig:k2k3}, we have depicted the first $\mathcal{P}$-positions of $K^2$ and $K^3$. Green (resp. red) squares are $\mathcal{P}$-positions of $K^3$ (resp. $K^2$). The common $\mathcal{P}$-positions are colored in blue.
\begin{figure}[h!tb]
  \centering
  \scalebox{.5}{\includegraphics{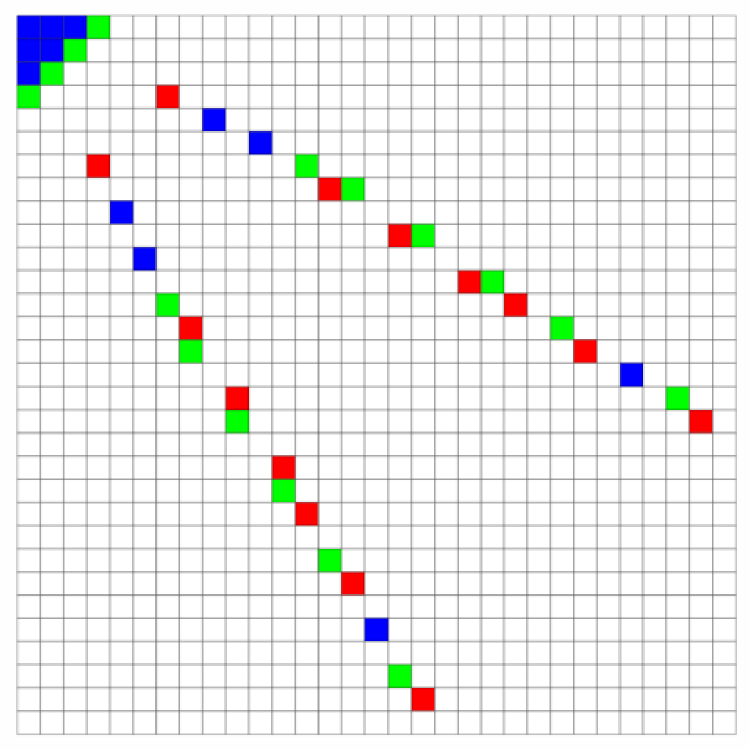}}
  \caption{Comparing $\mathcal{P}$-positions of $K^2$ (red) and $K^3$ (green).}
  \label{fig:k2k3}
\end{figure}

The first non-terminal $\mathcal{P}$-positions $(a_n,b_n)_{n\ge 0}$ are given in Table~\ref{tab:k3}. From now on, we take this convention: in the indexing of the $P$-positions, we do not take terminal positions into account. Thus, $(a_0,b_0)$ represents the first non-terminal $\mathcal{P}$-position. Due to the symmetry of the game, we only consider sequences where $a_n<b_n$. 
\begin{table}[h!t]
\[
  \begin{array}{c|cccccccccccccccccccccccccccccc}
    & 0 & 1 & 2 & 3 & 4 & 5 & 6 & 7 & 8 & 9 & 10 & 11 & 12 & 13 & 14 & 15 & 16 & 17 \\
    \hline
 a_n& 4 & 5 & 6 & 7 & 9 & 11 & 13 & 15 & 16 & 18 & 19 & 21 & 22 & 24 & 25 & 27 &
   29 & 30 \\
 b_n & 8 & 10 & 12 & 14 & 17 & 20 & 23 & 26 & 28 & 31 & 33 & 36 & 38 & 41 & 43 & 46
   & 49 & 51 & \\
  \end{array}
  \]
  \caption{First (non-terminal) $\mathcal{P}$-positions of $K^3$, $(n\ge 0)$.}
  \label{tab:k3}
\end{table}
\begin{theorem}\label{thm:fibaut2}
There exists a Fibonacci-automatic function $g_3:\mathbb{N}\to\{0,1,2\}$ such that, for all $n\ge 0$, the $n$th non-terminal $\mathcal{P}$-position $(a_n,b_n)$ of $K^3$ is given by 
\begin{equation}\label{eq:K3}
  a_n=\lfloor (n+2)\phi \rfloor+g_3(n+1)-1 \quad \text{ and }\quad 
  b_n=\lfloor (n+2)\phi^2 \rfloor+g_3(n+1)+1.
\end{equation}
\end{theorem}
The first few values of $g_3(n)$ are
\[
  1221011221211111111121121121211211112111121111111211112\cdots
\]
Using the heuristic described in Section~\ref{ssec:heuristic}, assuming that the above infinite word is generated by a $\varphi$-morphism, we obtain the morphism 
\begin{align*}
  \mu_3:\;&
0 \mapsto  01,
1 \mapsto  2,
2 \mapsto  34,
3 \mapsto  56,
4 \mapsto  7,
5 \mapsto  78,\\ 
& 6 \mapsto  9,
7 \mapsto  ab,
8 \mapsto  a,
9 \mapsto 56,
a \mapsto ab, 
b \mapsto  7
\end{align*}
and the coding
\[\rho_3:0,3,5,6,8,a,b\mapsto 1,\quad 1,2,7,9\mapsto 2,\quad 4\mapsto 0.\]
We thus define $g_3(n)$ as the $n$th symbol in $\rho_3(\mu_3^\omega(1))$. Note that to produce enough value of $g_3(n)$, we have to compute sufficiently many $\mathcal{P}$-positions of $W^3$.

\begin{proof}
Since we have a morphism~$\mu_3$ and a coding~$\rho_3$, we can store in {\tt Word Automata Lib} a text file {\tt K3.txt} describing a $12$-state DFAO. The transitions are given by $\mu_3$ and the outputs by $\rho_3$. Since we have a candidate set in the statement, \eqref{eq:K3} can be encoded as follows.
\begin{verbatim}
def pposK3_asym "?msd_fib (a+b<=3) | (En,x $phin(n+2,x)
                     & a+1=x+K3[n+1] & b=x+n+3+K3[n+1])":
def pposK3 "?msd_fib $pposK3_asym(a,b) | $pposK3_asym(b,a)":
\end{verbatim}
Thanks to Proposition~\ref{pro:kernel}, we simply have to check stability and absorption of this set.
\begin{verbatim}
eval stableK3 "?msd_fib Ap,q,r,s (($pposK3(p,q) & $pposK3(r,s)
  & p>=r & q>=s & p+q>3) => ((p=r & q=s) | (p>r & q>s & p+s!=q+r)))":
eval absorbingK3 "?msd_fib Ap,q (~$pposK3(p,q) => Ex,y
        (x<=p & y<=q & $pposK3(x,y) & (p+y=q+x | p=x | q=y))) ":
\end{verbatim}
Both commands return {\tt True} (intermediate computations need at most 449 states).
\end{proof}

\subsection{Even further}
For the game~$K^4$, we have applied the same procedure. The first non-terminal $\mathcal{P}$-positions $(a_n,b_n)_{n\ge 0}$ are given in Table~\ref{tab:k4}.
\begin{table}[h!t]
\[
\begin{array}{c|cccccccccccccccccc}
  & 0 & 1 & 2 & 3 & 4 & 5 & 6 & 7 & 8 & 9 & 10 & 11 & 12 & 13 & 14 & 15 & 16 & 17 \\
  \hline
 a_n &5 & 6 & 7 & 8 & 9 & 11 & 13 & 15 & 17 & 19 & 20 & 22 & 23 & 25 & 26 & 28 & 29 & 31 \\
 b_n & 10 & 12 & 14 & 16 & 18 & 21 & 24 & 27 & 30 & 33 & 35 & 38 & 40 & 43 & 45 & 48 & 50 & 53 \\
\end{array}
  \]
  \caption{First (non-terminal) $\mathcal{P}$-positions of $K^4$, $(n\ge 0)$.}
  \label{tab:k4}
\end{table}
To get a long enough prefix, we have computed the first $600$ $\mathcal{P}$-positions $(a_n,b_n)$ with $a_n\le b_n$ (up to $a_n$ close to $1000$). From this, we conjectured that
\begin{equation}\label{eq:K4}
  a_n=\lfloor (n+2)\phi \rfloor+g_4(n+1) \quad \text{ and }\quad 
  b_n=\lfloor (n+2)\phi^2 \rfloor+g_4(n+1)+3,
\end{equation}
where the first few values of $g_4(n)$ are
\[
  12210001112111111010110110111111111111111111\cdots.
\]
Having the first $600$ values of $g_4(n)$, we apply the heuristic of Section~\ref{ssec:heuristic}. Note that here, to be able to discern types correctly, we have to consider $4$-types (there are distinct $4$-types whose restrictions to $3$-types are equal). Since we have to take fourth iterations, this explains why we need a longer prefix. We obtain $18$ different $4$-types and the heuristic gives the following morphisms:
\begin{align*}
  \mu_4:\;&0 \mapsto 01,\ 1 \mapsto 2,\ 2 \mapsto 34,\ 3 \mapsto 56,\ 4 \mapsto 7,\ 5 \mapsto 89,\ 6 \mapsto a,\ 7 \mapsto bc,\ 8 \mapsto dc,\\
   &9 \mapsto d,\ a \mapsto ef,\ b \mapsto ef,\ c \mapsto e,\ d \mapsto 7g,\ e \mapsto eh,\ f \mapsto b,\ g \mapsto d,\ h \mapsto b  
\end{align*}
\[
  \rho_4: 0,3,7,8,9,b,c,d,e,h \mapsto 1,\  1,2,a \mapsto 2,\  4,5,6,f,g \mapsto 0
\]
With $g_4(n)$ being the $n$th symbol of $\rho_4(\mu_4^\omega(1))$, we have a $18$-state DFAO. Formulas \eqref{eq:K4} translates in {\tt Walnut} into
\begin{verbatim}
def pposK4_asym "?msd_fib (a+b<=4) | (En,x $phin(n+2,x) & (a=x+K4[n+1])
                                                        & (b=x+K4[n+1]+n+5))":
\end{verbatim}
and we then follow exactly the same proof to get the result.

\begin{theorem}\label{thm:fibaut3}
  The $\mathcal{P}$-positions $(a,b)$, with $a\le b$, of the game $K^4$ with a set of terminal positions being $\{(x,y)\mid x+y\le 4\}$ and Wythoff's moves are characterized by \eqref{eq:K4}. 
\end{theorem}

Thus, aside from the fact that the morphisms $\mu_\ell$ and $\rho_\ell$ used to describe the set of $\mathcal{P}$-positions are defined over increasingly large alphabets --- which requires computing a sufficiently long prefix --- these preliminary results suggest that such a result could be extended to any arbitrary value of the parameter~$\ell$. In Section~\ref{sec:6}, we give a recursive characterization of the $\mathcal{P}$-positions for all values of the parameter $\ell\ge 1$. From that result, we are able to give with Corollary~\ref{cor:behave} the general behavior of the sequences $(a_n)_{n\ge 0}$ and $(b_n)_{n\ge 0}$.


\section{A recursive characterization for $K^\ell$}\label{sec:6}

It is classical in the combinatorial game literature \cite{Duchene,Fraenkel3,Fraenkel1,Fraenkel2} to have several characterizations of the set of $\mathcal{P}$-positions (algebraic, recursive, combinatorial, morphic, \ldots). In this section, we provide a recursive characterization of the set of $\mathcal{P}$-positions of $K^\ell$, for all $\ell\ge 1$.

First, we consider small values of the parameter ($\ell=1$ and $\ell=2$) and give an automatic proof. It is then not difficult to state a general conjecture that we are able to prove by classical methods (stability and absorption). Indeed, it is not possible to obtain an automated parameterized proof. Here we clearly see how the two approaches complement each other: {\tt Walnut} allows us to conjecture general statements that are then proved in the classical way.

Let us recall our indexing convention: $(a_0,b_0)$ is the first non-terminal $\mathcal{P}$-position of $K^\ell$ (with $a_0<b_0$). It is indeed not necessary to encode the terminal positions, which are set in advance and known to both players.

Let us begin with some classical and well-known results on Wythoff game before moving on to the games discussed in this article. We recall the {\em recursive definition} of the $\mathcal{P}$-positions of Wythoff game~$K^0$, with our convention:
\[
  (a_0,b_0)=(1,2) \quad\text{ and }\quad \forall n\ge 1, \left\{
    \begin{array}{l}
a_n =\mex (\{ a_i,b_i \mid i<n \}\cup\{0\}),\\
b_n =a_n+n+1,\\
    \end{array}
  \right.
\]
where $\mex$ stands for {\em minimum excluded value}, i.e. the least non-negative integer not in the set. In other words, $\mex S=\min(\mathbb{N}\setminus S)$. So $a_1=\mex\{0,1,2\}=3$ and $b_1=a_1+2=5$, then $a_2=\mex\{0,1,2,3,5\}=4$ and $b_2=a_2+3=7$ and so on.


We now turn to the games discussed in this paper. For $K^1$ (see Table~\ref{tab:K1}), we have the following result.
\begin{proposition}
  Define recursively the sequence $(a_n,b_n)_{n\ge 0}$ by
\[
  (a_0,b_0)=(2,4) \quad\text{ and }\quad \forall n\ge 1, \left\{
    \begin{array}{l}
a_n =\mex (\{ a_i,b_i \mid i<n \}\cup\{0,1\}),\\
b_n =a_n+n+2.\\
    \end{array}
  \right.
\]
The set non-terminal $\mathcal{P}$-positions of $K^1$ is $\{(a_n,b_n)\mid n\ge 0\} \cup \{(b_n,a_n)\mid n\ge 0\} $.
\end{proposition}

\begin{proof}
It is clear from Theorem~\ref{thm:carK1} that the sequences $(a_n)_{n\ge 0}=2,3,5,7,\ldots$ and $(b_n)_{n\ge 0}=4,6,9,12,\ldots$ are a partition of $\mathbb{N}_{\ge 2}$. Indeed, any integer has a unique representation in the Fibonacci system and each representation ends either with $0$ or $01$. We first define two predicates respectively for the sets of pairs $(a_n,n)$ and $(b_n,n)$. We simply make use of the shift (observe that, in Table~\ref{tab:K1}, suppressing the last zero from every representation in the first column gives the set of all valid Fibonacci representations). We use $n+1$ instead of $n$ because of our indexing convention.  
\begin{verbatim}
def pposK1_an "?msd_fib $shift(n+1,a)":
\end{verbatim}
As an example, we have
\begin{verbatim}
eval test "?msd_fib $pposK1_an(2,0) & $pposK1_an(3,1) &
                    $pposK1_an(5,2) & $pposK1_an(7,3)":
\end{verbatim}
and similarly, we shift twice and add $1$ to get the last column in Table~\ref{tab:K1},
\begin{verbatim}
def pposK1_bn "?msd_fib Et,c $shift(n+1,t) & $shift(t,c) & b=c+1":
\end{verbatim}
with
\begin{verbatim}
eval test "?msd_fib $pposK1_bn(4,0) & $pposK1_bn(6,1) &
                    $pposK1_bn(9,2) & $pposK1_bn(12,3)":
\end{verbatim}
Since we have a partition, we only have to check that, for all $n\ge 0$, $b_n=a_n+n+2$, which is readily verified with
\begin{verbatim}
eval test "?msd_fib An Ea,b $pposK1_an(a,n) & $pposK1_bn(b,n) & b=a+n+2":
\end{verbatim}
\end{proof}

For $K^2$, see Figure~\ref{fig:k2k3}, we get the following recursive characterization.
\begin{proposition} 
Define recursively the sequence $(a_n,b_n)_{n\ge 0}$ by
\[
  (a_0,b_0)=(3,6) \quad\text{ and }\quad \forall n\ge 1, \left\{
    \begin{array}{l}
a_n =\mex (\{ a_i,b_i \mid i<n \}\cup\{0,1,2\}),\\
b_n =a_n+n+3.\\
    \end{array}
  \right.
\]
The set non-terminal $\mathcal{P}$-positions of $K^2$ is $\{(a_n,b_n)\mid n\ge 0\} \cup \{(b_n,a_n)\mid n\ge 0\} $.
\end{proposition}

\begin{proof}
The sequence $(a_n)_{n\ge 0}$ and $(b_n)_{n\ge 0}$ are Fibonacci-synchronized (see \cite[Sec.~10.11]{ShallitBook})
\begin{verbatim}
def pposK2_an "?msd_fib Ex ($phin(n+2,x) & a+1=x+G[n+2])":
\end{verbatim}
So this predicates defines the pairs $(a_n,n)$ with $n\ge 0$. As an example, the following command returns {\tt True}:
\begin{verbatim}
eval test "?msd_fib $pposK2_an(3,0) & $pposK2_an(4,1)
                  & $pposK2_an(5,2) & $pposK2_an(7,3)":
\end{verbatim}
Similarly, the following predicates defines the pairs $(b_n,n)$ with $n\ge 0$.
\begin{verbatim}
def pposK2_bn "?msd_fib Ex $phin(n+2,x) & b=x+n+2+G[n+2]":
eval test "?msd_fib $pposK2_bn(6,0) & $pposK2_bn(8,1)
                  & $pposK2_bn(10,2) & $pposK2_bn(13,3)":
\end{verbatim}
The difference of indices with Section~\ref{ssec:autom_proof} comes from our indexing convention.
The next two commands show that the two sequences $(a_n)_{n\ge 0}$ and $(b_n)_{n\ge 0}$ make a partition of $\mathbb{N}_{\ge 3}$:
\begin{verbatim}
eval covering "?msd_fib Ax (x>=3 => En ($pposK2_an(x,n) | $pposK2_bn(x,n)))":
eval notboth "?msd_fib ~(Em,n,x ($pposK2_an(x,m) & $pposK2_bn(x,n)))":
\end{verbatim}
The first one means that every $x\ge 3$ belongs to at least of the two sequences. The second one means that there is no $x\ge 3$ that is simultaneously of the form $a_m$ and $b_n$. Showing the partition could have been defined with the predicate \verb!pposK2_asym!. Synchronization is only needed to show that $b_n-a_n-n=3$. So, finally, we can verify that if $a_n=x$, then $b_n=x+n+3$.
\begin{verbatim}
eval alinkb "?msd_fib  Ax (x>=3 => (En ($pposK2_an(x,n)
                              => $pposK2_bn(x+n+3,n))))":
\end{verbatim}
This could also be immediately derived from the expression \eqref{eq:Ppos}. The fact that $b_n=a_n+n+3$ and that the sequences $(a_n)_{n\ge 0}$ and $(b_n)_{n\ge 0}$ make a partition of $\mathbb{N}_{\ge 3}$ imply the $\mex$ definition of the $a_n$'s (roughly speaking, gaps are filled with $a$'s).
\end{proof}

Now that we have inspected the cases $K^1$ and $K^2$, the next statement seems natural.

\begin{theorem}\label{thm:rec_car}
  Let $\ell\ge 1$ and define recursively $(a_n,b_n)_{n\ge 0}$ by
\begin{equation}\label{eq:mexdef}
  (a_0,b_0)=(\ell+1,2\ell+2) \quad\text{ and }\quad \forall n\ge 1, \left\{
    \begin{array}{l}
a_n =\mex (\{ a_i,b_i \mid i<n \}\cup\{0,1,\ldots,\ell\}),\\
b_n =a_n+n+\ell+1.\\
    \end{array}
  \right. 
\end{equation} 
The set of non-terminal $\mathcal{P}$-positions of $K^\ell$ is $\{(a_n,b_n)\mid n\ge 0\} \cup \{(b_n,a_n)\mid n\ge 0\} $.
\end{theorem}

We will repeatedly use the following properties all along the proofs of the next three results.

\begin{lemma}\label{lem:anbn}
  For the sequences \eqref{eq:mexdef} defined above, we have the following properties.
  \begin{enumerate}
  \item The sequences $(a_n)_{n\ge 0}$ and $(b_n)_{n\ge 0}$ are strictly increasing.
  \item They make a partition of $\mathbb{N}_{> \ell}$. 
  \item For each $j>\ell$, there exist a unique $n$ such that $j=b_n-a_n$.
  \item For all $n$, $a_{n+1}-a_n\in\{1,2\}$ and $b_{n+1}-b_n\in\{2,3\}$. In particular, $b_{n+1}-b_n=a_{n+1}-a_n+1$.
  \end{enumerate}
\end{lemma}

\begin{proof}
  These properties obviously follow from the $\mex$ definition, we provide a proof for the sake of completeness.

  From the $\mex$ definition, each $a_n$ is a new integer strictly larger than
  $a_{n-1}$, so $(a_n)_{n\ge 0}$ is strictly increasing. Moreover
  $b_n=a_n+\ell+1+n>b_{n-1}$, hence $(b_n)_{n\ge 0}$ is strictly increasing and
  the two sequences are disjoint.

  Suppose some $j>\ell$ never appears in the two sequences and take the smallest
  such $j$. Since $(a_n)_{n\ge 0}$ is increasing, let $n$ be minimal such that
  $a_n> j$. By minimality of $j$, all integers in $\{\ell+1,\ldots,j-1\}$ 
  already appear among $a_i,b_i$ for $i<n$, so $j$ is the smallest
  unused integer at step $n$, forcing $a_n=j$, a contradiction. Hence
  $(a_n)_{n\ge 0}$ and $(b_n)_{n\ge 0}$ form a partition of $\mathbb{N}_{> \ell}$.

  The identity
  $b_n-a_n=n+\ell+1$ shows that for each $j>\ell$ there is a unique
  $n$ with $j=b_n-a_n$. Finally, since $a_{n+1}$ is the $\mex$ of the
  previous values, it cannot exceed $a_n+2$, for otherwise either
  $a_n+1$ or $a_n+2$ would be unused and smaller; thus
  $a_{n+1}-a_n\in\{1,2\}$. From $b_n=a_n+\ell+1+n$ we obtain
  $b_{n+1}-b_n=(a_{n+1}-a_n)+1\in\{2,3\}$.

\end{proof}

\begin{proof}[Proof of Theorem~\ref{thm:rec_car}]
  Let $\mathcal{P}$ be the union of set  $\{(a_n,b_n)\mid n\ge 0\} \cup \{(b_n,a_n)\mid n\ge 0\} $ and the set of terminal positions $\{(x,y)\mid x+y\le\ell\}$. We use Proposition~\ref{pro:kernel} and prove the stability and absorption of $\mathcal{P}$.

  {\em Stability}. Let $(x,y)$ be in $\mathcal{P}$. We may assume that $x+y>\ell$ because there is no allowed move from a terminal position. So let $(x,y)$ be of the form $(a_n,b_n)$ for some $n\ge 0$. We have to show that every move leads to a position not in $\mathcal{P}$. The case $(b_n,a_n)$ is symmetric and not discussed below.
  \begin{itemize}
  \item A vertical move will lead to a position $(a_n,y')$ with $0\le y'<b_n$. Since $a_n\ge \ell+1$, for all $n\ge 0$, $(a_n,y')$ is not a terminal position: $a_n+y'>\ell$. To get a contradiction, assume that $(a_n,y')$ is in $\mathcal{P}$. Then $(a_n,y')$ is either of the form $(a_m,b_m)$ or $(b_m,a_m)$.

    \begin{itemize}
    \item If $(a_n,y')=(a_m,b_m)$, then $a_n=a_m$ and since the sequence $(a_n)_{n\ge 0}$ is increasing, we get $n=m$ and thus $y'=b_n$ contradicting $y'<b_n$.
    \item If $(a_n,y')=(b_m,a_m)$ then $a_n=b_m$ but this contradicts the fact that
\[
  \{a_n\mid n\ge 0\}\cap \{b_n\mid n\ge 0\}=\emptyset.
\]
\end{itemize}

\item The proof is the same for a horizontal move leading to $(x',b_n)$ with $0\le x'<a_n$.

\item Now consider a diagonal move leading to a position $(a_n-t,b_n-t)$ for some $t\in\{1,\ldots,a_n\}$. First, such a position is not a terminal one because
\[
  a_n-t+b_n-t=a_n+b_n-2t\ge a_n+b_n-2a_n = n+\ell+1 >\ell.
\]

Assume that $(a_n-t,b_n-t)$ is of the form $(a_m,b_m)$, then
\[
  b_m-a_m = b_n-t -(a_n-t)= n+\ell+1
\]
and thus, $m=n$ implying $t=0$, a contradiction.

Finally, assume that $(a_n-t,b_n-t)$ is of the form $(b_m,a_m)$, then
\[
  b_m-a_m = a_n-t -(b_n-t) = - (n+\ell+1) <0,
\]
which is again a contradiction because $b_m-a_m=m+\ell+1>0$.
\end{itemize}

{\em Absorption}. Let $(x,y)$ a position not in $\mathcal{P}$. In particular, $x+y>\ell$. We have to prove that there is a move from $(x,y)$ into $\mathcal{P}$. By symmetry, we may assume that $x\le y$.
\begin{itemize}
\item Case 1: $y-x\le\ell$. We will prove that there is a diagonal move leading to a terminal position. Consider the quantity
  \[
    t=\left\lceil \frac{x+y-\ell}{2}\right\rceil. 
  \]
  Since $x+y>\ell$, $t\ge 1$. Note that
  \[
    x+y-\ell=2x+ (y-x) -\ell\le 2x.
  \]
  Hence $t\le x$ and $(x,y)\to(x-t,y-t)$ is a legal move. To conclude with this part, $(x-t,y-t)$ is a terminal move:
  \begin{align*}
    x-t+y-t\le x+y - (x+y-\ell) = \ell.
  \end{align*}
  
\item Case 2: $y-x>\ell$. There exists a unique $n\ge 0$ such that $y-x=n+\ell+1=b_n-a_n$.
  \begin{itemize}
  \item If $x\ge a_n$, consider $t=x-a_n$. If $t>0$, we can play the diagonal move $(x,y)\to(x-t,y-t)=(a_n,y-x+a_n)=(a_n,b_n)$. If $t=0$, then $x=a_n$ and thus $y=b_n$, but $(x,y)$ is not in $\mathcal{P}$. So this case does not occur.
    
  \item If $x \le \ell<a_n$, then a vertical move $(x,y)\to (x,0)$ leads to a terminal position.
  \item If $\ell< x<a_n$, then $x$ is of the form $a_m$ or $b_m$ for some $m$ (because we have a partition of $\mathbb{N}_{>\ell}$). If $x=a_m$, since the sequence $(a_n)_{n\ge 0}$ is increasing $m<n$. So $y>b_m$ because $y-x=b_n-a_n>b_m-a_m$. A vertical move $(x,y)\to (x,b_m)=(a_m,b_m)$ is thus enough. Finally, if $x=b_m$, observe that $y\ge x = b_m >a_m$, so consider the move $(x,y)\to (x,a_m)=(b_m,a_m)$.
  \end{itemize}

\end{itemize}
\end{proof}

We now show that the points $(a_n,b_n)$ remain at bounded distance from a line of slope~$\phi$ (see, for instance, Figure~\ref{fig:k2k3}). In the proof, recursively indexed sequences of the form $x_{x_n}$ will play an important role. Note that the solution of Wythoff satisfies the complementary equation $x_{x_n}=y_n-1$, see \cite{Kimberling1,Kimberling2,Kimberling3}.  Our approach is in the same spirit as found in \cite{Fraenkel5} where uniform bounds for $\mex$-defined complementary sequences are derived.

  \begin{corollary}\label{cor:behave}
  For the sequences defined by \eqref{eq:mexdef}, 
  there exist a bounded function $\lambda_\ell:\mathbb{N}\to\mathbb{Z}$ such that, for all $n$,   
  \[
   a_n=\lfloor (n+\ell)\phi \rfloor +\lambda_\ell(n).
 \]
 So, in particular, \(b_n=\lfloor (n+\ell)\phi^2 \rfloor + \lambda_\ell(n) +1\).
  \end{corollary}

  \begin{proof}
    Let
    \[
      S_n:=\sum_{k=1}^n (a_k-a_{k-1}-1).
    \]
    Since $a_k-a_{k-1}=2$ if and only if some $b_j$ is inserted between $a_{k-1}$ and $a_k$, we have that
    \[
      S_n=\#\{i\mid b_i<a_n\},
    \]
    that is $b_0<b_1<\cdots <b_{S_n-1}<a_n<b_{S_n}$. We set $S_0=0$ and $S_1=\cdots=S_\ell=0$, $S_{\ell+1}=1$ and $(S_n)_{n\ge 0}$ is non-decreasing.
    
    The interval $[a_0,a_n]$ contains $n+1$ terms $a_j$'s and $S_n$ terms $b_j$'s, i.e., $n+1+S_n=a_n-a_0+1$. Since $a_0=\ell+1$, we have
    \begin{equation}\label{eq:anSn}
      a_n=\ell+1+n+S_n.
    \end{equation}
    By definition of $S_n$, we have
    \begin{align*}
      & b_{S_n-1}<a_n<b_{S_n} \\
      \Leftrightarrow\; & a_{S_n-1}+S_n-1+\ell+1 < a_n < a_{S_n}+S_n+\ell+1\\
      \Leftrightarrow\; & a_{S_n-1}-1 < a_n-S_n-\ell-1 < a_{S_n}\\
      \Leftrightarrow\; & a_{S_n-1} \le a_n-S_n-\ell-1 < a_{S_n}.
    \end{align*}
    Now using \eqref{eq:anSn},
      \begin{align*}
        & \ell+1+S_n-1+S_{S_n-1}  \le n < \ell+1+S_n+S_{S_n} \\
        \Leftrightarrow\; &  S_{S_n-1}-1 \le n -\ell-1-S_n < S_{S_n}
      \end{align*}
      Note that $S_{S_n}-S_{S_{n-1}}=a_{S_n}-a_{S_n-1}-1\in\{0,1\}$. So, if $S_{S_n}-S_{S_{n-1}}=0$, then  the integer $n -\ell-1-S_n$ belongs to $[S_{S_n}-1,S_{S_n})$, and is thus equal to $S_{S_n}-1$. If $S_{S_n}-S_{S_{n-1}}=1$, then  $n -\ell-1-S_n$ belongs to $[S_{S_n}-2,S_{S_n})$. Consequently,
      \[
        n -\ell-1-S_n= S_{S_n}-1-\epsilon_n, \text{ for some }\epsilon_n\in\{0,1\}.
      \]
      That is,
      \[
       S_n+ S_{S_n}= n -\ell+\epsilon_n.
      \]

      We want to prove that $|a_n-n \phi|$ is bounded by a constant. Thanks to \eqref{eq:anSn}, observe that
      \[
        a_n-n \phi= \ell+1+S_n-(\phi-1)n = \ell+1+S_n- \frac{n}{\phi}.
      \]
      So it is enough to bound $|D_n|$, where
      \[
        D_n:= S_n-\frac{n}{\phi}.
        \]
      We have 
      \[
        S_n+ S_{S_n}= D_n +\frac{n}{\phi} +D_{S_n}+\frac{S_n}{\phi}= n -\ell+\epsilon_n
      \]
      and
      \[
        \frac{S_n}{\phi}=\frac{1}{\phi} \left(D_n+\frac{n}{\phi}\right).
      \]
      Putting these together, we get
      \[
        \left( 1+\frac{1}{\phi}\right) D_n + D_{S_n} + n \left( \frac{1}{\phi}+\frac{1}{\phi^2}\right)= n -\ell+\epsilon_n.
      \]
      Hence,
      \[
        \phi  D_n = -\ell +\epsilon_n - D_{S_n}
      \]
      and
      \begin{equation}\label{eq:bDN}
        \phi  \left| D_n \right|  \le \ell + \left|  D_{S_n}\right|.
      \end{equation}
      We claim that, for all $n$,
       \[
         \left| D_n\right| \le \frac{\ell}{\phi-1}=\phi\ell.
       \]
       Note that for $i\le \ell$, $|D_i|=i/\phi<\ell/(\phi-1)=\phi\ell$. 
       We proceed by contradiction. Assume that there exists a minimal $N>\ell$ such that $|D_N|>\phi\ell$. We have $S_N<N$, so by minimality of~$N$,
       \[
         \left|D_{S_N}\right| \le \phi\ell.
       \]
       From \eqref{eq:bDN}, we have
       \[
           \left| D_N \right|  \le \frac{\ell}{\phi} + \frac{\left|D_{S_N}\right|}{\phi} \le  \frac{\ell}{\phi} + \frac{\phi\ell}{\phi}=\left(\frac{1}{\phi}+1\right)\ell=\phi\ell,
         \]
         which is a contradiction. Consequently, using \eqref{eq:anSn},
         \[
           a_n-\lfloor (n+\ell)\phi\rfloor = a_n- (n+\ell)\phi + \{(n+\ell)\phi \} = \ell (1-\phi) +1+ D_n+\{(n+\ell)\phi \},
         \]
         which is a bounded function. The form of $b_n$ comes again from the fact that $b_n=a_n+n+\ell+1$.
  \end{proof}

The following statement is an immediate consequence of the previous corollary. However, we provide a proof because we develop arguments (about density) different from those used earlier. This result is in fact weaker than the previous one: the asymptotic behavior described below is that $a_n-n\phi=o(n)$. However, it does not yield a bounded discrepancy $a_n-n\phi=O(1)$ as in Corollary~\ref{cor:behave}: the points $(a_n,b_n)$ remain at bounded distance from a line of slope~$\phi$.
  
\begin{corollary}
  For the sequences defined by \eqref{eq:mexdef}, the following limits exist:
\[
  \lim_{n\to\infty} \frac{a_n}{n}=\phi\quad\text{ and }\quad  \lim_{n\to\infty} \frac{b_n}{n}=\phi^2.
\]
\end{corollary}

\begin{proof}
  Let
  \[
    \pi_A(x):=\#\{n\mid a_n<x\} \quad\text{ and }\quad \pi_B(x):=\#\{n\mid b_n<x\}.
  \]
  Since $a_0<\cdots <a_{n-1}<a_n$, we have $\pi_A(a_n)=n$ and similarly, $\pi_B(b_n)=n$. Since the two sequences make a partition of $\mathbb{N}_{>\ell}$, for every integer $x>\ell+1$, we have
  \begin{equation}
    \label{eq:total}
  \pi_A(x)+\pi_B(x)=x-\ell-1.  
  \end{equation}
  Particularly, $ \pi_A(b_n)+\pi_B(b_n)=b_n-\ell-1$. Using $\pi_B(b_n)=n$, we get\footnote{With the notation of the proof of Corollary~\ref{cor:behave}, $S_n=\pi_B(a_n)$.}
  \[
    \pi_A(b_n)=b_n-n-\ell-1=a_n.
  \]
  
  Consider
  \[
    \alpha:=\limsup_{n\to \infty} \frac{a_n}{n} \quad\text{ and }\quad \beta:=\liminf_{n\to \infty} \frac{a_n}{n}.
  \]
  From Lemma~\ref{lem:anbn}, we have $2\ge \alpha\ge\beta\ge 1$. Our aim is to show that $\alpha=\beta$.

  The intervals $[a_n,a_{n+1})$ make a partition of $\mathbb{N}_{>\ell}$. If $x\in [a_n,a_{n+1})$, then $\pi_A(x)\in\{n,n+1\}$ and
  \[
    \frac{n+1}{a_{n+1}}\cdot\underbrace{\frac{n}{n+1}}_{\to 1}=\frac{n}{a_{n+1}} < \frac{\pi_A(x)}{x} \le \frac{n+1}{a_n}=\frac{n}{a_n}\cdot\underbrace{\frac{n+1}{n}}_{\to 1}.
  \]
  From this, we deduce that
  \[
    \limsup_{x\to\infty} \frac{\pi_A(x)}{x} =\limsup_{n\to \infty}\frac{n}{a_n}=\frac{1}{\liminf_{n\to\infty}\frac{a_n}{n}}=\frac{1}{\beta}.
    \]

    Similarly, the intervals $[b_n,b_{n+1})$ make a partition of $\mathbb{N}_{>2\ell+1}$. If $x\in [b_n,b_{n+1})$, then 
  \[
    \frac{a_n}{b_{n}}\cdot\underbrace{\frac{b_n}{b_{n+1}}}_{\to 1}=\frac{\pi_A(b_n)}{b_{n+1}} \le \frac{\pi_A(x)}{x} \le \frac{\pi_A(b_{n+1})}{b_n}=\frac{a_{n+1}}{b_{n+1}}\cdot\underbrace{ \frac{b_{n+1}}{b_n}}_{\to 1}.
  \]
  So
   \[
    \limsup_{x\to\infty} \frac{\pi_A(x)}{x} =\limsup_{n\to \infty}\frac{a_n}{b_n}.
  \]
  Let $g(x)=\frac{x}{x+1}$. First observe that
  \[
    \frac{a_n}{b_n}=\frac{\frac{a_n}{n}}{\frac{a_n}{n}+1+\frac{\ell+1}{n}}.
  \]
Second, 
  \[
    \left| \frac{a_n}{b_n} - g\left(\frac{a_n}{n}\right)\right| =
    \left| \frac{\frac{a_n}{n}}{\frac{a_n}{n}+1+\frac{\ell+1}{n}}-g\left(\frac{a_n}{n}\right) \right| = \left| \frac{\frac{a_n}{n} (-\frac{\ell+1}{n})}{\left( \frac{a_n}{n}+1+\frac{\ell+1}{n}\right) \left( \frac{a_n}{n}+1\right)} \right|\to 0.
  \]
  The function $g$ is continuous and strictly increasing on $[0,+\infty)$. Hence,
   \[
   \frac{1}{\beta}= \limsup_{x\to\infty} \frac{\pi_A(x)}{x} =\limsup_{n\to \infty}\frac{a_n}{b_n}=  \limsup_{n\to \infty} g\left(\frac{a_n}{n}\right) = g\left( \limsup_{n\to \infty} \frac{a_n}{n}\right)=g(\alpha)=\frac{\alpha}{\alpha+1}.
  \]
  Proceeding in the same way with $\liminf$, we get
  \[
    \frac{1}{\alpha}=\frac{\beta}{\beta+1}.
  \]
  Hence, $\alpha+1=\beta \alpha=\beta+1$ and therefore, $\alpha=\beta$ meaning that the sequence $(a_n/n)_{n\ge 1}$ is converging to $\alpha$. Note that
  \[
    \frac{b_n}{n}=\frac{a_n}{n}+1+\frac{\ell+1}{n}\to \alpha+1, \quad \text{  if }n\to\infty.
  \]
  The proof also shows that $(\pi_A(n)/n)_{n\ge 1}$ converges to $1/\alpha$. Proceeding as before,   the intervals $[b_n,b_{n+1})$ make a partition of $\mathbb{N}_{>2\ell+1}$. If $x\in [b_n,b_{n+1})$, then $\pi_B(x)\in\{n,n+1\}$ and
  \[
    \frac{n+1}{b_{n+1}}\cdot\underbrace{\frac{n}{n+1}}_{\to 1}=\frac{n}{b_{n+1}} < \frac{\pi_B(x)}{x} \le \frac{n+1}{b_n}=\frac{n}{b_n}\cdot\underbrace{\frac{n+1}{n}}_{\to 1}.
  \]
  From this, we deduce that
  \[
    \limsup_{x\to\infty} \frac{\pi_B(x)}{x} =\limsup_{n\to \infty}\frac{n}{b_n}=\frac{1}{\liminf_{n\to\infty}\frac{b_n}{n}}=\frac{1}{\alpha+1}.
    \]
    Since
    \[
    \frac{\pi_A(x)+\pi_B(x)}{x}\to 1\text{ (from \eqref{eq:total})} \quad\text{and}\quad \frac{\pi_A(x)}{x}\to\frac{1}{\alpha},
    \]
    it follows that $\frac{\pi_B(x)}{x}$ is converging to $\frac{1}{\alpha+1}$ as $x\to\infty$. We thus get
    \[
      1=\frac{1}{\alpha}+\frac{1}{\alpha+1}.
    \]
    Since $1\le\alpha\le 2$, we obtain that the limit $\alpha$ is the golden ratio $\phi=(1+\sqrt{5})/2$.
  \end{proof}

  \begin{remark}
    One may also ask whether or not, the sequences $(a_n)_{n\ge 0}$ and $(b_n)_{n\ge 0}$ studied here can be complementary Beatty sequences. In \cite{Bosh}, it is shown that the prefix $P=(c_n)_{1\le n\le L}$ of length~$L$ a sequence $(c_n)_{n\ge 1}$ is a {\em spectrum}, i.e., of the form $\lfloor n\alpha+\beta\rfloor$ for some reals $\alpha,\beta$, if and only if $\underline{d}(P)<\overline{d}(P)$ where
    \[
      \underline{d}(P)=\max_{1\le i<k\le L}\frac{c_k-c_{k-i}-1}{i}\quad\text{ and }\quad
      \overline{d}(P)=\max_{1\le i<k\le L}\frac{c_k-c_{k-i}+1}{i}.\]
    For $\ell\ge 2$, the prefix of length~$L$ of the corresponding sequence $(a_n)_{n\ge 0}$ is such that
    \[
      \underline{d}(P)=\frac{2\ell+1}{\ell+1}>\overline{d}(P)=\frac{\ell+1}{\ell}.
    \]
    Hence $(a_n)_{n\ge 0}$ is not a spectrum. 

    For the game $K^1$, the situation is different; we have (with our convention where $(a_0,b_0)=(2,4)$)
    \[
      a_n=\lfloor (n+1)\phi -1\rfloor \quad\text{ and }\quad b_n=\lfloor (n+1)\phi^2 -1\rfloor.
    \]
    For more on complementary sequences, we refer the reader to \cite{Fraenkel4,Larsson}.
  \end{remark}


\section{A morphic characterization for $K^1, K^2, K^3,\ldots$}\label{sec:morphic}

The {\em morphic characterization}, as introduced in \cite{Duchene}, is given by the Fibonacci word. If we start indexing letters of the Fibonacci word~$\mathbf{f}$ with $1$, the indices of the letters $a$ (resp. $b$) in~$\mathbf{f}$ correspond to the sequence $(a_n)_{n\ge 0}$ (resp. $(b_n)_{n\ge 0}$)
\[
  \begin{array}{c|ccccccccccc}
    n&1&2&3&4&5&6&7&8&9&10&\cdots \\
    \hline
\mathbf{f}&a&b&a&a&b&a&b&a&a&b&\cdots \\    
  \end{array}
\]
We say that $\mathbf{f}$ {\em codes} the $\mathcal{P}$-positions: the indices of the $n$th letter $a$ and the $n$th letter $b$ in $\mathbf{f}$ give $(a_{n-1},b_{n-1})$, $n\ge 1$. The reader can notice that the first $a$ and $b$ occur in positions $(1,2)$, the second $a$ and $b$ in positions $(3,5)$ and so on. For more on morphic characterization, see \cite[Sec. 3.2]{Vision} where Wythoff game is coined as the ``Fibonacci game'', showing the interplay between combinatorial games and combinatorics on words.

In this section, the various morphisms $f_\ell$ and $g_\ell$ have been obtained with the heuristic of Section~\ref{ssec:heuristic}. We also use the commands \verb!pposK1_an!, \verb!pposK1_bn!, \verb!pposK2_an! and \verb!pposK2_bn! introduced at the beginning of Section~\ref{sec:6}.

\begin{proposition}
   The non-terminal $\mathcal{P}$-positions $(a_n,b_n)_{n\ge 0}$ of $K^1$ are coded by the morphic word $g_1(f_1^\omega(0))$ where, indexing of the letters starts with $2$ and $f_1$ is a $\varphi$-morphism over a $5$-letter alphabet:
\begin{align*}
  f_1:\;&0 \mapsto 01, 1 \mapsto 2, 2 \mapsto 34, 3 \mapsto 31, 4 \mapsto 2,\\
  g_1:\;&0, 1, 3 \mapsto a,\quad 2, 4 \mapsto b.
\end{align*}
 \end{proposition}
\[
  \begin{array}{c|cccccccccccc}
    n&2&3&4&5&6&7&8&9&10&11& \cdots \\
    \hline
f_1^\omega(0)&0 & 1 & 2 & 3 & 4 & 3 & 1 & 2 & 3 & 1& \cdots\\
g_1(f_1^\omega(0))&a & a & b & a & b & a & a & b & a & a & \cdots\\
  \end{array}
\]

\begin{proof}
   We turn $f_1$ and $g_1$ into an automaton stored in {\tt Word Automata Library} where outputs are $1$ and $2$ instead of $a$ and $b$. Let $\mathbf{w}$ be the corresponding infinite word. We have to check that for all $a_n$ (resp. $b_n$), the corresponding symbol in $\mathbf{w}$ is $1$ (resp. $2$). Since {\tt Walnut} is indexing infinite words from $0$, we must have
  \[
 \forall n\ge 3,\quad   \mathbf{w}_{a_n-2}=1\ \text{ and }\ \mathbf{w}_{b_n-2}=2.
\]
That is verified thanks to
  \begin{verbatim}
eval matcha "?msd_fib  An (Ex (($pposK1_an(x,n)) => (F1G1[x-2]=1)))":
eval matchb "?msd_fib  An (Ex (($pposK1_bn(x,n)) => (F1G1[x-2]=2)))":
\end{verbatim}
\end{proof}

\begin{proposition}
 The non-terminal $\mathcal{P}$-positions $(a_n,b_n)_{n\ge 0}$ of $K^2$ are coded by the morphic word $g_2(f_2^\omega(0))$, where indexing of the letters starts with $3$ and $f_2$ is a $\varphi$-morphism over a $16$-letter alphabet:
\begin{align*}
  f_2:\;& 0 \mapsto 01, 1 \mapsto 2, 2 \mapsto 34, 3 \mapsto 56, 4 \mapsto 7, 5 \mapsto 89, 6 \mapsto (10), 7 \mapsto (11)(12), 8 \mapsto (10)(13),\\
    & 9 \mapsto (14), 10 \mapsto (10)(13), 11 \mapsto 56, 12 \mapsto (15), 13 \mapsto 5, 14 \mapsto 89, 15 \mapsto (11)(12),\\
  g_2:\;& 0, 1, 2, 4, 6, 8, 9, 11, 12, 13, 14, 15 \mapsto a,\quad 3, 5, 7, 10 \mapsto b.
\end{align*}
\end{proposition}
\[
  \begin{array}{c|ccccccccccccccccccccc}
    n&3&4&5&6&7&8&9&10&11&12&13&14&15&16&17&18&19&20&\cdots \\
    \hline
f_2^\omega(0)& 0 & 1 & 2 & 3 & 4 & 5 & 6 & 7 & 8 & 9 & 10 & 11 & 12 & 10 & 13 & 14 & 10 & 13 & \cdots \\
g_2(f_2^\omega(0))& a & a & a & b & a & b & a & b & a & a & b & a & a & b & a & a & b & a & \cdots \\
  \end{array}
\]
\begin{proof}
  Same proof as the previous one. We turn $f_2$ and $g_2$ into an automaton and check
  \begin{verbatim}
eval matcha "?msd_fib  An (Ex (($pposK2_an(x,n)) => (F2G2[x-3]=1)))":
eval matchb "?msd_fib  An (Ex (($pposK2_bn(x,n)) => (F2G2[x-3]=2)))":
\end{verbatim}
\end{proof}


One proceed in the same way for the next value of $\ell$.

\begin{proposition}
 The non-terminal $\mathcal{P}$-positions $(a_n,b_n)_{n\ge 0}$ of $K^3$ are coded by the morphic word $g_3(f_3^\omega(0))$, where indexing of the letters starts with $4$ and $f_3$ is a $\varphi$-morphism over a $22$-letter alphabet:
\begin{align*}
  f_3:\;& 0 \mapsto 01, 1 \mapsto 2, 2 \mapsto 34, 3 \mapsto 56, 4 \mapsto 7, 5 \mapsto 89, 6 \mapsto (10), 7 \mapsto (11)(12),  8 \mapsto (13)(12), \\
      & 9 \mapsto (14), 10 \mapsto (15)(16), 11 \mapsto (14)(17), 12 \mapsto (18), 13 \mapsto (14)(17), 14 \mapsto (19)(12),  \\
      & 15 \mapsto (18)(20), 16 \mapsto (21), 17 \mapsto (18), 18 \mapsto (13)(12), 19 \mapsto (19)(12), 20 \mapsto (14), 21 \mapsto (15)(16),\\
g_3:\;& 0, 1, 2, 3, 5, 7, 9, 11, 12, 14, 16, 18, 21 \mapsto a,\quad 4, 6, 8, 10, 13, 15, 17, 19, 20 \mapsto b.\\
\end{align*}
\end{proposition}

\begin{remark}
  When applying the heuristic, to get $f_2$ and $g_2$ right (resp. $f_3$ and $g_3$), we had to consider $4$-types (resp. $5$-types).
\end{remark}


\section{Blocking Wythoff}\label{sec:3}

Larsson characterized \cite[Thm.~1.1]{block} the set of $\mathcal{P}$-positions of $W^2$, a variant of Wythoff game where one option may be blocked, see Figure~\ref{fig:w2w3}. We provide an automatic proof of this result. 
\begin{theorem}
The set of $\mathcal{P}$-positions of $W^2$ is given by
\[
  \mathcal{R}_2=\{(0,0)\}\cup \{ \{n,2n+1\}\mid n\ge 0\}\cup \{\{2\lfloor n\phi\rfloor+2,2\lfloor n\phi^2\rfloor+2\}\mid n\ge 0\}.
\]  
\end{theorem}

\begin{definition}\label{def:optW}
It is convenient to define the options of a position (for Wythoff's moves) by
\begin{verbatim}
def optW "?msd_fib (c=a & d<b) | (c<a & d=b) | (a+d=c+b & c<a)":
\end{verbatim}
This means that $(a,b,c,d)$ belongs to the set if and only if there is a Wythoff move from the position $(a,b)$ to $(c,d)$, i.e., $(c,d)$ is an option of $(a,b)$.
\end{definition}

\begin{proof}
The set $\mathcal{R}_2$ can be defined by
\begin{verbatim}
def pposW2_asym "?msd_fib En,a,b ((x=0 & y=0) | (x=n & y=2*n+1) |
                ($phin(n,a) & $phi2n(n,b) & x=2*a+2 & y=2*b+2))":
def pposW2 "?msd_fib $pposW2_asym(a,b) | $pposW2_asym(b,a)":
\end{verbatim}
requiring a $73$ states automaton.

We again use Proposition~\ref{pro:kernel}. Clearly, for each $\mathcal{P}$-position $(p,q)$, at most one of its options $(r,s)$ is a $\mathcal{P}$-position (if this is the case, the previous player will forbid this option). The following formula expresses that, if there is indeed such a $\mathcal{P}$-position $(r,s)$, then for any other option $(a,b)$, if it is also a $\mathcal{P}$-position, then it must coincide with $(r,s)$.
\begin{verbatim}
eval stable "?msd_fib Ap,q ($pposW2(p,q) =>
     (Er,s ($pposW2(r,s) & $optW(p,q,r,s)) =>
     (Aa,b ($optW(p,q,a,b) & $pposW2(a,b)) => (a=r&b=s))))":
\end{verbatim}
intermediate computations require $6452$ states and the result is {\tt TRUE}.

Similarly, for each $\mathcal{N}$-position $(p,q)$, there is at at least
two $\mathcal{P}$-positions $(r,s)\neq (t,u)$ in its set of options. So that the previous player cannot forbid both of these. This is expressed as
\begin{verbatim}
eval absorbing "?msd_fib Ap,q (~$pposW2(p,q) =>
                 Er,s,t,u ($optW(p,q,r,s) & $optW(p,q,t,u)
                 & (r!=t|s!=u) & $pposW2(r,s) & $pposW2(t,u)))":
\end{verbatim}
intermediate computations require $34813$ states and return {\tt TRUE}.
\end{proof}

We also provide an automatic proof of the following result from \cite{block}.
\begin{theorem}
The set of $\mathcal{P}$-positions of $W^3$ where $2$ options may be blocked is given by
\[
  \mathcal{R}_3=\{(0,0)\}\cup \{ \{n,2n+1\},\{n,2n+2\}\mid n\ge 0\}.
\]  
\end{theorem}
\begin{proof}
The set $\mathcal{R}_2$ can be defined by
\begin{verbatim}
def pposW3_asym "?msd_fib En ((x=0 & y=0) | (x=n & y=2*n+1) | (x=n&y=2*n+2))":
def pposW3 "?msd_fib $pposW3_asym(a,b) | $pposW3_asym(b,a)":
\end{verbatim}
requiring $28$ states.

The strategy is the same as for $W^2$ using Proposition~\ref{pro:kernel}. For each $\mathcal{P}$-position $(p,q)$, at most two of its options $(r,s)$ and $(t,u)$ are $\mathcal{P}$-positions (here, the previous player may forbid two options). The following formula expresses that if there is indeed such $\mathcal{P}$-positions, then for any other option $(a,b)$, if it is also a $\mathcal{P}$-position, then it must coincide with $(r,s)$ or $(t,u)$.
\begin{verbatim}
eval stableW3 "?msd_fib Ap,q ($pposW3(p,q) =>
      (Er,s,t,u ($pposW3(r,s) & $pposW3(t,u) & $optW(p,q,r,s) & $optW(p,q,t,u)
      & (r!=t|s!=u)) =>
      (Aa,b ($optW(p,q,a,b) & $pposW3(a,b)) => ((a=r&b=s)|(a=t&r=u)))))":
\end{verbatim}
Similarly, for each $\mathcal{N}$-position $(p,q)$, there is at at least
three $\mathcal{P}$-positions in its set of options. So that the previous player cannot forbid all of these. This is expressed as
\begin{verbatim}
eval absorbW3 "?msd_fib Ap,q (~$pposW3(p,q) =>
             Er,s,t,u,v,w ($optW(p,q,r,s) & $optW(p,q,t,u) & $optW(p,q,v,w)
             & (r!=t|s!=u) & (r!=v|s!=w) & (v!=t|w!=u)
             & $pposW3(r,s) & $pposW3(t,u) & $pposW3(v,w)))":
\end{verbatim}
Both commands evaluate to {\tt TRUE}, intermediate computations requiring at most $116141$ states.
\end{proof}

\begin{figure}[h!t]
  \centering
  \scalebox{.7}{\includegraphics{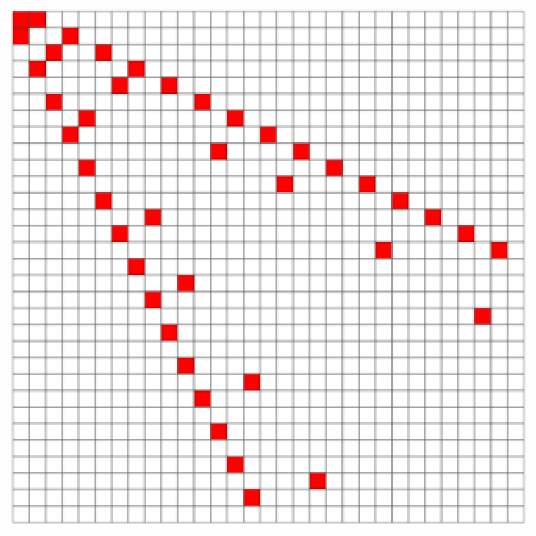}}
    \scalebox{.7}{\includegraphics{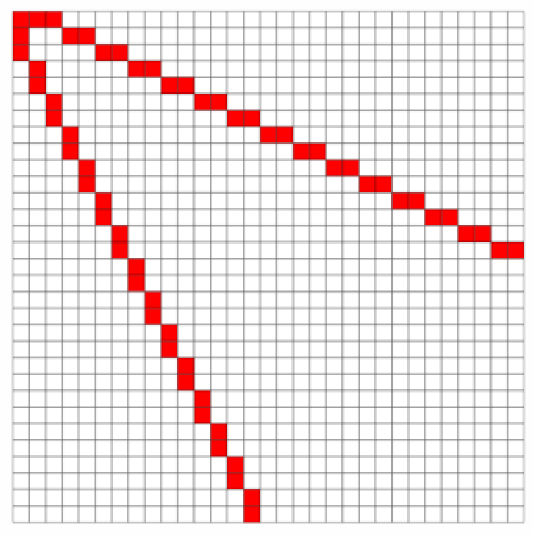}}
  \caption{The $\mathcal{P}$-positions of $W^2$ and $W^3$.}
  \label{fig:w2w3}
\end{figure}

\section{About redundant moves}\label{sec:redundant}

A classical question about game is the following. Can we modify the rule-set in such a way that the set of $\mathcal{P}$-positions remains the same? One can add any move provided that it does not permit to move between two $\mathcal{P}$-positions breaking the stability of the set. So we focus on the question of removing a move. With a smaller rule-set, stability trivially still holds. The question thus concerns absorption. Let us recall what a redundant move is.
\begin{definition}
A move of a game $G$ is said to be \emph{redundant} if removing it from the rule-set does not change the set of $\mathcal{P}$-positions of $G$.
\end{definition}

It is shown in \cite{Nowakowski2010} that the classical Wythoff's game has no redundant move. This result has also been verified automatically with \texttt{Walnut} in \cite{MRRW}. Using a similar approach, we show that the variations of Wythoff game considered in this paper have no redundant move.

\subsection{No redundant moves in $K^\ell$, $\ell\le 4$}

We follow the procedure described in \cite{MRRW}. For each move $m$, that is, for all pairs of the set
\[
\bigcup_{i>0}\;\{(i,0),(0,i),(i,i)\},
\]
one needs to check that there exists an $\mathcal{N}$-position $(p,q)$ such that  $m$ is the only winning move leading to a $\mathcal{P}$-position. This shows that $m$ is necessary to preserve the absorbing property, and thus the set of $\mathcal{P}$-positions.

This can be expressed in {\tt Walnut} using the following three commands (one for each type of move). We make use of Definition~\ref{def:optW} for the options {\tt optW}. We illustrate this on the game~$K^2$:

\begin{verbatim}
eval nr_1 "?msd_fib Ai (i>0 => (Ep,q ~$pposK2(p,q) & $pposK2(p-i,q)
    & (Ar,s ((r!=p-i & s!=q & $optW(p,q,r,s)) => ~$pposK2(r,s)))))": 

eval nr_2 "?msd_fib Ai (i>0 => (Ep,q ~$pposK2(p,q) & $pposK2(p,q-i)
    & (Ar,s ((r!=p & s!=q-i & $optW(p,q,r,s)) => ~$pposK2(r,s)))))": 

eval nr_3 "?msd_fib Ai (i>0 => (Ep,q ~$pposK2(p,q) & $pposK2(p-i,q-i)
    & (Ar,s ((r!=p-i & s!=q-i & $optW(p,q,r,s)) => ~$pposK2(r,s)))))":
\end{verbatim}
All three commands evaluate to {\tt TRUE}, with intermediate steps computing automata with up to $2130$ states.

\begin{theorem}
  Let $\ell\in\{1,2,3,4\}$. The game $K^\ell$ has no redundant move, i.e., removing any move from the rule-set change the set of $\mathcal{P}$-positions.  
\end{theorem}

\begin{proof}
  Simply repeat the above three commands, replacing the predicate {\tt pposK2} conveniently. The longest computation is for \verb!nr_3! and $K^4$, which requires up to $5140$ states.
\end{proof}
\subsection{No redundant moves in $W^2$ and $W^3$}

Let us first consider the game $W^2$. Here, we have to take into account the blocking maneuver: since the opponent can forbid one option, a move $m$ is non-redundant if there exists an $\mathcal{N}$-position $(p,q)$ with exactly two options in $\mathcal{P}$ that can be reached using two moves $m$ and $m'$ respectively. 

Indeed, if $m$ is withdrawn from the rule-set, then $m'$ will be the only option remaining to go from $(p,q)$ to a $\mathcal{P}$-position. Since the previous player may prevent the next one from playing along this move, the absorbing property is not satisfied, and $m$ is thus necessary to preserve the set of $\mathcal{P}$-positions.

Again, this can be expressed in {\tt Walnut}. Let us start with a move of the form $(i,0)$. The first formula defines the set of pairs $(i,j)$ such that there exists an $\mathcal{N}$-position $(p,q)$ from which there are exactly two distinct moves of the form $(i,0)$ and $(j,0)$, with $i\neq j$, leading to $\mathcal{P}$-positions. Any option $(r,s)$ from $(p,q)$ which is a $\mathcal{P}$-position is thus reached using one of these two. If a pair belongs to this set, then $(i,0)$ and $(j,0)$ are non-redundant. The second (resp. third) formula handles the situation where there exists an $\mathcal{N}$-position $(p,q)$ from which exactly two moves of the form $(i,0)$ and $(0,j)$ (resp. $(i,0)$ and $(j,j)$) are leading to $\mathcal{P}$-positions. Here, there is no restriction on $i$ and $j$ because we use different types of moves.
\begin{verbatim}
def opt2W2a "?msd_fib i!=j & Ep,q (p>=i & p>=j & ~$pposW2(p,q)
                          & $pposW2(p-i,q) & $pposW2(p-j,q)
                         & Ar,s (($optW(p,q,r,s) & $pposW2(r,s))
                           => ((r=p-i& s=q) | (r=p-j& s=q))))":

def opt2W2b "?msd_fib  Ep,q (p>=i & q>=j & ~$pposW2(p,q)
                          & $pposW2(p-i,q) & $pposW2(p,q-j)
                         & Ar,s (($optW(p,q,r,s) & $pposW2(r,s))
                           => ((r=p-i& s=q) | (r=p& s=q-j))))":

def opt2W2c "?msd_fib  Ep,q (p>=i & p>=j & q>=j & ~$pposW2(p,q)
                          & $pposW2(p-i,q) & $pposW2(p-j,q-j)
                         & Ar,s (($optW(p,q,r,s) & $pposW2(r,s))
                           => ((r=p-i& s=q) | (r=p-j& s=q-j))))":
\end{verbatim}
Intermediate computations for each of these formulas require around $11000$ states. Now, we evaluate
\begin{verbatim}
eval opt2W2 "?msd_fib Ai (i>0=> (Ej ($opt2W2a(i,j)|$opt2W2b(i,j)
                                                 |$opt2W2c(i,j))))":
\end{verbatim}
which returns {\tt True}. This means that $(i,0)$ is a non-redundant move, for all $i>0$. By symmetry, this is also the case for a move $(0,i)$.

We now consider moves of the form $(i,i)$. The strategy is similar. We just have to ensure in the third formula that the moves $(i,i)$ and $(j,j)$ are distinct with $i\neq q$.
\begin{verbatim}
def opt2W2d "?msd_fib  Ep,q (p>=i & q>=i & p>=j & ~$pposW2(p,q)
                          & $pposW2(p-i,q-i) & $pposW2(p-j,q)
                         & Ar,s (($optW(p,q,r,s) & $pposW2(r,s))
                           => ((r=p-i& s=q-i) | (r=p-j& s=q))))":

def opt2W2e "?msd_fib  Ep,q (p>=i & q>=i & q>=j & ~$pposW2(p,q)
                          & $pposW2(p-i,q-i) & $pposW2(p,q-j)
                         & Ar,s (($optW(p,q,r,s) & $pposW2(r,s))
                           => ((r=p-i& s=q-i) | (r=p& s=q-j))))":

def opt2W2f "?msd_fib  i!=j & Ep,q (p>=i & q>=i & ~$pposW2(p,q)
                          & $pposW2(p-i,q-i) & $pposW2(p-j,q-j)
                         & Ar,s (($optW(p,q,r,s) & $pposW2(r,s))
                           => ((r=p-i& s=q-i) | (r=p-j& s=q-j))))":
\end{verbatim}
We evaluate
\begin{verbatim}
eval opt2W2 "?msd_fib Ai (i>0 => (Ej ($opt2W2d(i,j)|$opt2W2e(i,j)
                                                  |$opt2W2f(i,j))))":
\end{verbatim}
returning {\tt True} and consequently, we have the following result.
\begin{theorem}
  The game $W^2$ (resp. $W^3$) has no redundant move, i.e., removing any move from the rule-set change the set of $\mathcal{P}$-positions.
\end{theorem}

For the sake of readability, the commands for $W^3$ are given in appendix.


\section{Conclusions}

The deep connections existing between combinatorial games and numeration systems have long been exploited by many authors \cite{DFNR,Duchene,Fraenkel3,Fraenkel1,Fraenkel2}, particularly to determine in polynomial time whether a given position is a $\mathcal{P}$-position. Wythoff game and many of its variants are closely related to the Fibonacci numeration system. In this article, following the line of our previous work \cite{MRRW}, we take these ideas one step further providing more evidence. Using a decidable theory implemented in the freely available software {\tt Walnut}, we have obtained many results with little effort! These examples are particularly well-suited because the Fibonacci numeration system admits a simple automaton recognizing addition. This ensures that the automata produced during intermediate formula-reduction steps hopefully remain of manageable size.

The starting point of this article was the rather unexpected connection between Hofstadter’s $G$-sequence \seqnum{A005206} and the combinatorial game $K^2$ introduced by Komak et {\em al.} \cite{Komak}. This led us to generalize $K^2$ to a family of games $K^\ell$. It would therefore be interesting to relate this family to the generalized Hofstadter sequences studied in \cite{Letouzey}, for which connections with numeration systems are of interest \seqnum{A005374}, \seqnum{A005375}, \seqnum{A005376}, \seqnum{A100721}. This could also be linked to Narayana representations as discussed in \cite{Nara1}. In particular, can one define an analogue of the function $g$ from Proposition~\ref{prop:h} and Theorem~\ref{thm:K2} for other values of the parameter~$\ell$? We did not pursue this direction, as our aim was to obtain Fibonacci-automatic functions $g_2$, $g_3$, $g_4$ with Proposition~\ref{pro:fibaut} and Theorems~\ref{thm:fibaut2} and \ref{thm:fibaut3}.

On a different aspect, although there is no unique logical formula that defines properties such as stability or absorption, one can observe in our automatic proofs that the order of magnitude of the automata constructed by {\tt Walnut} differs between the games $K^\ell$ (for $\ell\le 4$) with a few thousand states on the one hand and the games with blocking maneuvers $W^2$ and $W^3$ on the other, with more than $10^5$ states. Thus, one could in some sense quantify --- assuming a notion of minimal representation/logical formula can be defined --- the complexity of a combinatorial game by the size of the automata involved in the proof characterizing its $\mathcal{P}$-positions. Of course, this is reasonable for games having a description within the same numeration system, since the size of the automaton recognizing addition plays a role in these constructions.

\bigskip
\hrule
\bigskip

\noindent 2020 {\it Mathematics Subject Classification}: Primary 91A46. ; Secondary 11A67, 68Q45, 68V15.

\noindent \emph{Keywords:} Combinatorial games ; Wythoff game ; Walnut theorem-prover ; first-order logic.

\bigskip
\hrule
\bigskip

\noindent (Concerned with sequences
\seqnum{A005206},
\seqnum{A000201},  \seqnum{A022342}, \seqnum{A003622} and
\seqnum{A001950}
.)

\section{Appendix}
Here are the commands showing that $W^3$ has no redundant move. With the command \verb!opt2W3_mn!, we consider a move of the first kind $(i,0)$ and two moves of kind \verb!m! and \verb!n! ; the three being pairwise distinct. Moves of the second (resp. third) kind are of the form $(0,i)$ and $(i,i)$. When at least two moves are of the same kind, we have to ensure that they are distinct.

\begin{verbatim}
def opt2W3_11 "?msd_fib i!=j & j!=k & i!=k & Ep,q (p>=i & p>=j & p>=k
 & ~$pposW3(p,q) & $pposW3(p-i,q) & $pposW3(p-j,q) & $pposW3(p-k,q)
 & Ar,s (($optW(p,q,r,s) & $pposW3(r,s))
 => ((r=p-i & s=q) | (r=p-j & s=q) | (r=p-k & s=q))))":

def opt2W3_12 "?msd_fib i!=j & Ep,q (p>=i & p>=j & q>=k & ~$pposW3(p,q)
 & $pposW3(p-i,q) & $pposW3(p-j,q) & $pposW3(p,q-k)
 & Ar,s (($optW(p,q,r,s) & $pposW3(r,s))
 => ((r=p-i & s=q) | (r=p-j & s=q) | (r=p & s=q-k))))":

def opt2W3_22 "?msd_fib j!=k & Ep,q (p>=i & q>=j & q>=k & ~$pposW3(p,q)
 & $pposW3(p-i,q) & $pposW3(p,q-j) & $pposW3(p,q-k)
 & Ar,s (($optW(p,q,r,s) & $pposW3(r,s))
 => ((r=p-i & s=q) | (r=p & s=q-j) | (r=p & s=q-k))))":

def opt2W3_13 "?msd_fib i!=j & Ep,q (p>=i & p>=j & p>=k & q>=k
 & ~$pposW3(p,q) & $pposW3(p-i,q) & $pposW3(p-j,q) & $pposW3(p-k,q-k)
 & Ar,s (($optW(p,q,r,s) & $pposW3(r,s))
 => ((r=p-i & s=q) | (r=p-j & s=q) | (r=p-k & s=q-k))))":

def opt2W3_23 "?msd_fib Ep,q (p>=i & q>=j & p>=k & q>=k & ~$pposW3(p,q)
 & $pposW3(p-i,q) & $pposW3(p,q-j) & $pposW3(p-k,q-k)
 & Ar,s (($optW(p,q,r,s) & $pposW3(r,s))
 => ((r=p-i & s=q) | (r=p & s=q-j) | (r=p-k & s=q-k))))":

def opt2W3_33 "?msd_fib j!=k & Ep,q (p>=i & p>=j & q>=j & p>=k & q>=k
 & ~$pposW3(p,q) & $pposW3(p-i,q) & $pposW3(p-j,q-j) & $pposW3(p-k,q-k)
 & Ar,s (($optW(p,q,r,s) & $pposW3(r,s))
 => ((r=p-i & s=q) | (r=p-j & s=q-j) | (r=p-k & s=q-k))))":
\end{verbatim}
and finally, we check:
\begin{verbatim}
eval redundantW3_1 "?msd_fib Ai (i>0
 => Ej,k ($opt2W3_11(i,j,k) | $opt2W3_12(i,j,k) | $opt2W3_22(i,j,k) |
          $opt2W3_13(i,j,k) | $opt2W3_23(i,j,k) | $opt2W3_33(i,j,k)))":
\end{verbatim}
By symmetry, we do not have to consider a first move of the second kind. But we still have to consider a first move of the third kind. 
\begin{verbatim}
def opt3W3_11 "?msd_fib j!=k & Ep,q (p>=i & q>=i & p>=j & p>=k
 & ~$pposW3(p,q) & $pposW3(p-i,q-i) & $pposW3(p-j,q) & $pposW3(p-k,q)
 & Ar,s (($optW(p,q,r,s) & $pposW3(r,s))
 => ((r=p-i & s=q-i) | (r=p-j & s=q) | (r=p-k & s=q))))":

def opt3W3_12 "?msd_fib Ep,q (p>=i & q>=i & p>=j & q>=k
 & ~$pposW3(p,q) & $pposW3(p-i,q-i) & $pposW3(p-j,q) & $pposW3(p,q-k)
 & Ar,s (($optW(p,q,r,s) & $pposW3(r,s))
 => ((r=p-i & s=q-i) | (r=p-j & s=q) | (r=p & s=q-k))))":

def opt3W3_22 "?msd_fib j!=k & Ep,q (p>=i & q>=i & q>=j & q>=k
 & ~$pposW3(p,q) & $pposW3(p-i,q-i) & $pposW3(p,q-j) & $pposW3(p,q-k)
 & Ar,s (($optW(p,q,r,s) & $pposW3(r,s))
 => ((r=p-i & s=q-i) | (r=p & s=q-j) | (r=p & s=q-k))))":

def opt3W3_13 "?msd_fib i!=k & Ep,q (p>=i & q>=i & p>=j & p>=k & q>=k
 & ~$pposW3(p,q) & $pposW3(p-i,q-i) & $pposW3(p-j,q) & $pposW3(p-k,q-k)
 & Ar,s (($optW(p,q,r,s) & $pposW3(r,s))
 => ((r=p-i & s=q-i) | (r=p-j & s=q) | (r=p-k & s=q-k))))":

def opt3W3_23 "?msd_fib i!=k & Ep,q (p>=i & q>=i & q>=j & p>=k & q>=k
 & ~$pposW3(p,q) & $pposW3(p-i,q-i) & $pposW3(p,q-j) & $pposW3(p-k,q-k)
 & Ar,s (($optW(p,q,r,s) & $pposW3(r,s))
 => ((r=p-i & s=q-i) | (r=p & s=q-j) | (r=p-k & s=q-k))))":

def opt3W3_33 "?msd_fib i!=j & j!=k & i!=k & Ep,q
 (p>=i & q>=i & p>=j & q>=j & p>=k & q>=k & ~$pposW3(p,q)
 & $pposW3(p-i,q-i) & $pposW3(p-j,q-j) & $pposW3(p-k,q-k)
 & Ar,s (($optW(p,q,r,s) & $pposW3(r,s))
 => ((r=p-i & s=q-i) | (r=p-j & s=q-j) | (r=p-k & s=q-k))))":

eval redundantW3_2 "?msd_fib Ai (i>0
 => Ej,k ($opt3W3_11(i,j,k) | $opt3W3_12(i,j,k) | $opt3W3_22(i,j,k) |
          $opt3W3_13(i,j,k) | $opt3W3_23(i,j,k) | $opt3W3_33(i,j,k)))":
\end{verbatim}

\end{document}